\documentclass[acmsmall,nonacm]{acmart}
\newif\iffull
\fulltrue % appendix on
\startPage{1}
\usepackage{math-cmds}

\usepackage[normalem]{ulem}
\usepackage{latexsym,amsthm,amsmath,amsfonts,stmaryrd}
\usepackage{listings}
\usepackage{caption}
\usepackage{xcolor}
\usepackage{colortbl}
\usepackage{enumitem}

\usepackage{centernot}
\usepackage{scalerel}

\usepackage[export]{adjustbox}
\usepackage{wrapfig}

\usepackage{multirow}

\usepackage{newtxmath}
\usepackage{subfig}
\usepackage[utf8]{inputenc}
\usepackage[T1]{fontenc}
\usepackage{libertine, mathpartir}

\renewcommand{\infer}[3][]{
\ifthenelse{\equal{#1}{}}{
\inferrule{#2}{#3}
}
{
\inferrule*[right={\scriptsize \textbf{#1}}]
{#2}
{#3}
}
}

\usepackage{braket}
\newtheorem{theorem}{Theorem}

\newcommand{\myparagraph}[1]{{\bf {#1}}}

\newcommand{\ttt}[1]{\texttt{#1}}

\newcommand{\out}[1] {}

\newcounter{codeLineCntr}

\newcommand{\cfbox}[2]{%
    \colorlet{currentcolor}{.}%
    {\color{#1}%
    \fbox{\color{currentcolor}#2}}%
}

\reversemarginpar
\newif\ifnotes
\notestrue

\newcommand{\punt}[1]{}

\newcommand{\secref}[1]{Section~\ref{sec:#1}}

\newcommand{\appref}[1]{Appendix~\ref{app:#1}}
\newcommand{\figref}[1]{Figure~\ref{fig:#1}}

\newcommand{\lineref}[1]{line~\ref{line:#1}}

\newcommand{\thmref}[1]{Theorem~\ref{thm:#1}}

\newcommand{\lemref}[1]{Lemma~\ref{lem:#1}}

\newcommand{\corref}[1]{Corollary~\ref{cor:#1}}

\renewcommand{\eqref}[1]{Equation~(\ref{eq:#1})}

\newcounter{remark}[section]

\newcommand{\myremark}[3]{
\refstepcounter{remark}
\[
\color{red}
\left\{
\sf
\parbox{0.8\columnwidth}
{
{\bf {#1}'s remark~\theremark:}
{#3}
}
\right.
\]
\marginpar{\bf {#2}.~\theremark}
}

\newcommand{\sremark}[1]{\myremark{Sam}{S}{#1}}

\newcommand{\sr}[1]{\sremark{#1}}

\newcommand{\todo}[1]{}

\setcopyright{none}
\begin{document}

%\title{Scaling Quantum Circuit Optimizers}
%\title{Linear Time Optimization of Quantum Circuits}
%\title{Linear Time Optimization of Quantum Circuit with Local Optimality}
%\title{Local Optimality for Quantum Circuits in Linear Time}
%\title{Local Optimality in Quantum Circuits}
%\title{Local Optimization of Large Quantum Circuits}
\title{Local Optimization of Quantum Circuits}
\iffull
\title{Local Optimization of Quantum Circuits (Extended Version)}
\fi
%\title{Efficient Optimization of Quantum Circuits via Local Optimality}
\author{Jatin Arora}
\email{jatina@andrew.cmu.edu}
\affiliation{%
	\institution{Carnegie Mellon University}
	\city{Pittsburgh}
	\state{PA}
	\country{USA}
}
\author{Mingkuan Xu}
\email{mingkuan@cmu.edu}
\affiliation{%
	\institution{Carnegie Mellon University}
	\city{Pittsburgh}
	\state{PA}
	\country{USA}
}
\author{Sam Westrick}
\email{shw8119@nyu.edu}
\affiliation{%
	\institution{New York University}
	\city{New York}
	\state{NY}
	\country{USA}
}
\author{Pengyu Liu}
\email{pengyuliu@cmu.edu}
\affiliation{%
	\institution{Carnegie Mellon University}
	\city{Pittsburgh}
	\state{PA}
	\country{USA}
}
\author{Dantong Li}
\email{dantong.li@yale.edu}
\affiliation{%
	\institution{Yale University}
	\city{New Haven}
	\state{CT}
	\country{USA}
}
\author{Yongshan Ding}
\email{yongshan.ding@yale.edu}
\affiliation{%
	\institution{Yale University}
	\city{New Haven}
	\state{CT}
	\country{USA}
}
\author{Umut A. Acar}
\email{umut@cmu.edu}
\affiliation{%
	\institution{Carnegie Mellon University}
	\city{Pittsburgh}
	\state{PA}
	\country{USA}
}
\date{}

\thispagestyle{empty}
\newcommand{\kwcost}[1]{\mathbf{cost}\left(  {#1} \right)}
\newcommand{\circuitcon}[2]{{#1} + {#2}}

\newcommand{\bigomega}{\mathbf{\Omega}}

%% Semantics frame
\newcommand{\sfbox}[1]
{
\cfbox{blue}{#1}
}
\newcommand{\srule}{\vspace{2mm}\rule{\columnwidth}{1pt}\vspace{2mm}}

\newcommand{\lang}{\textsc{Laqe}}
%% General syntax
\newcommand{\dom}[1]{\mathop{\text{dom}}(#1)}
\newcommand{\codom}[1]{\mathop{\text{cod}}(#1)}
\newcommand{\kw}[1]{\mbox{\ttt{#1}}}
\newcommand{\cdparens}[1]{({#1})}
\newcommand{\cd}[1]{{\lstinline!#1!}}
\newcommand{\hmm}{\textsf{HMM}}
\newcommand{\rulename}[1]{\textsc{#1}}
\newcommand{\ruleref}[1]{Rule~\rulename{#1}}

\newcommand{\true}{\ensuremath{\kw{true}}}
\newcommand{\false}{\ensuremath{\kw{false}}}
\newcommand{\ttrue}{\kw{t}}
\newcommand{\ffalse}{\kw{f}}
\newcommand{\prog}{\ensuremath{P}}
\newcommand{\pred}{\ensuremath{\mathcal{P}}}
\newcommand{\predf}[2]{\ensuremath{\pred(#1,#2)}}
\newcommand{\defeq}{\triangleq}

\newcommand{\type}[2]{\ensuremath{#1 : #2}}
\newcommand{\typed}[4]{\ensuremath{#1 \vdash_{#2} \type{#3}{#4}}}

\newcommand{\btype}{\beta}
\newcommand{\utype}{\theta}
\newcommand{\val}{v}
\newcommand{\uval}{u}

\newcommand{\algname}{\textsf{OAC}}
\newcommand{\algnameminus}{\textsf{OACMinus}}
\newcommand{\coam}{\textsf{OAC}}
\newcommand{\lopt}{\textsf{Lopt}}
\newcommand{\coamwith}[1]{\ensuremath{\mathsf{SOAM}[{#1}]}}
\newcommand{\queso}{{\textsf{Queso}}}
\newcommand{\voqc}{{\textsf{VOQC}}}
\newcommand{\pyzx}{{\textsf{PyZX}}}
\newcommand{\quartz}{{\textsf{Quartz}}}
\newcommand{\quartztool}{$\mathsf{Quartz}$}
\newcommand{\quesotool}{$\mathsf{Queso}$}
\newcommand{\feyntool}{\textsf{FeynOpt}}

\newcommand{\quartzt}[1]{\ensuremath{\mathsf{Quartz}_{\,#1}}}
\newcommand{\quesot}[1]{\ensuremath{\mathsf{Queso}_{\,#1}}}
\newcommand{\coamt}[1]{\ensuremath{\mathsf{SOAM}[#1]}}
\newcommand{\clifft}{Clifford+T}

\newcommand{\compat}[2]{\ensuremath{{#1} \mathbin{\scaleobj{1.2}{\diamond}} {#2}}}
\newcommand{\notcompat}[2]{\ensuremath{{#1} \mathbin{\scaleobj{1.2}{\centernot{\diamond}}} {#2}}}
\newcommand{\windowopt}[2]{{#2}~\textsf{\textbf{segment-optimal}}_{#1}}
\newcommand{\wopttext}{segment optimal}
\newcommand{\compressed}[1]{{#1}~\textsf{\textbf{compact}}}
\newcommand{\locallyopt}[2]{{#2}~\textsf{\textbf{locally-optimal}}_{#1}}
\newcommand{\qubits}[1]{\mathsf{qubits}({#1})}
%% Terms
\newcommand{\tmemprog}{memory-progress\xspace}
\newcommand{\tmempres}{memory-preservation\xspace}

%% imperative serial types
\newcommand{\kwint}{\kw{int}}
\newcommand{\kwnat}{\kw{nat}}
\newcommand{\kwfut}{\kw{fut}}
\newcommand{\kwprod}[2]{\ensuremath{{#1} \times {#2}}}
\newcommand{\kwarr}[2]{\ensuremath{{#1} \ra {#2}}}
\newcommand{\kwloc}[1]{\ensuremath{{#1}~\kw{loc}}}
\newcommand{\kwref}[1]{\ensuremath{{#1}~\kw{ref}}}

%% imperative multithreaded types
\newcommand{\kwtid}{\kw{tid}}
\newcommand{\kwunit}{\kw{unit}}
\newcommand{\kwok}{\kw{ok}}

% space, heap, store
\newcommand{\heap}{H}
\newcommand{\spc}{H}
\newcommand{\empspc}{\emptyset}
\newcommand{\spaceext}[3]{{#1}[{#2} \mapsto {#3}]}

\newcommand{\catspace}{\uplus}
\newcommand{\catheap}{\uplus}

\newcommand{\heapun}[1]{\langle #1 \rangle}
\newcommand{\heapbi}[2]{\langle #1 ; #2 \rangle}
\newcommand{\heaptri}[3]{\langle #1 ; #2; #3 \rangle}
\newcommand{\heapquad}[4]{\langle #1 ; #2 ; #3 ; #4 \rangle}
\newcommand{\restctx}[2]{\ensuremath{#1 \upharpoonright_{#2}}}
\newcommand{\freeloc}[1]{\ensuremath{\mathsf{FL}(#1)}}
\newcommand{\locs}[1]{\ensuremath{\mathsf{Loc}(#1)}}
\newcommand{\diff}[1]{\ensuremath{\mathsf{Diff}(#1)}}

\newcommand{\rename}[3]{[#2 \mapsto #3](#1)}
\newcommand{\kwt}{\kw{t}}

\newcommand{\estore}{[~]}
\newcommand{\mkstore}[2]{\ensuremath{{#1}::{#2}}}

%% imperative serial syntax
\newcommand{\kwn}{\kw{n}}
\newcommand{\kwlet}[3]{\kw{let}~{#1}={#2}~\kw{in}~{#3}~\kw{end}}
\newcommand{\kwfun}[3]{\ensuremath{\kw{fun}~{#1}~{#2}~\kw{is}~{#3}~\kw
{end}}}
\newcommand{\kwpair}[2]{\ensuremath{\langle{#1},{#2}}\rangle}
\newcommand{\kwapply}[2]{\ensuremath{{#1}~{#2}}}
\newcommand{\kwfst}[1]{\ensuremath{\kw{fst}\cdparens{#1}}}
\newcommand{\kwsnd}[1]{\ensuremath{\kw{snd}\cdparens{#1}}}
\newcommand{\gcing}[1]{\ensuremath{[#1]}}
\newcommand{\kwnew}[1]{\ensuremath{\kw{ref}(#1)}}
\newcommand{\kwderef}[1]{\ensuremath{\mathop{!}#1}}
\newcommand{\kwwrite}[2]{\ensuremath{#1 \mathop{:=} #2}}

\newcommand{\kwletrec}[2]{\ensuremath{#1 \mathop{\cdot} #2}}
\newcommand{\kwtask}[3]{\ensuremath{#1 \mathop{\cdot} #2 \mathop{\cdot} #3}}
\newcommand{\kwtaskalt}[2]{\ensuremath{#1 \mathop{\cdot} #2}}
\newcommand{\halt}{\bot}
\newcommand{\tree}{T}
\newcommand{\trace}{t}
%% imperative multithreaded syntax
\newcommand{\kwfork}[1]{{\ensuremath{\kw{fork}\cdparens{#1}}}}
\newcommand{\kwjoin}[1]{{\ensuremath{\kw{join}\cdparens{#1}}}}
\newcommand{\kwunitv}{\ensuremath{(\,)}}
\newcommand{\kwtidv}{\ensuremath{\kw{t}}}
\newcommand{\gheap}{G}
\newcommand{\lheap}{\heap}
\newcommand{\tolheap}[1]{\ensuremath{\Delta(#1)}}
\newcommand{\theap}[2]{\left(#1, #2\right)}

%% hierarchical syntax
\newcommand{\cdpar}{\texttt{par}}
\newcommand{\kwpar}[2]
           {\ensuremath{\mathop{\vartriangleleft \hspace{-0.1em} #1, #2
               \hspace{-0.1em} \vartriangleright}}}
\newcommand{\kwpara}[2]
           {\ensuremath{\mathop{\blacktriangleleft \hspace{-0.1em} #1, #2
               \hspace{-0.1em} \blacktriangleright}}}
\newcommand{\kwparl}[2]{\ensuremath{#1 \overset{\leftarrow}{\|} #2}}
\newcommand{\kwparr}[2]{\ensuremath{#1 \overset{\rightarrow}{\|} #2}}
\newcommand{\task}{T}
\newcommand{\config}{\mathcal{C}}

%% flattening
\newcommand{\flate}[1]{\hat{#1}}
\newcommand{\flatten}[3]{\left\| #2 \right\|_{#1} \leadsto #3}
\newcommand{\fstep}{\step}%{\step_F}

%% shorthands
\renewcommand{\a}{\ensuremath{\alpha}}
\renewcommand{\b}{\ensuremath{\beta}}
\newcommand{\h}{\ensuremath{\eta}}
\renewcommand{\r}{\ensuremath{\rho}}
\newcommand{\p}{\ensuremath{P}}
\newcommand{\s}{\p}
\newcommand{\om}{\ensuremath{\Omega}}
\renewcommand{\l}{\ensuremath{l}}
\newcommand{\sig}{\ensuremath{\Sigma}}
\newcommand{\empctx}{\ensuremath{\cdot}}

% Relations
%\newcommand{\red}{\Downarrow}
%\newcommand{\redgc}{\stackrel{gc?}{\Longrightarrow}}
%\newcommand{\alloc}{\stackrel{alloc}\Longrightarrow}
\newcommand{\la}{\leftarrow}
\newcommand{\ra}{\rightarrow}
\newcommand{\pstep}{\Rightarrow}
\newcommand{\tstep}{\Rightarrow}
\newcommand{\optstep}{\longmapsto}
\newcommand{\compstep}{\longmapsto_{\delta}}
\newcommand{\localstep}[3]{{#2} \overset{#1}{\optstep} {#3}}
\newcommand{\globstep}[2]{{#1} \compstep {#2}}
\newcommand{\step}{\ra}
\newcommand{\stepgc}[1]{\xra[{\mbox{\tiny GC}}]{#1}}
\newcommand{\gcstep}{\ra_{\mbox{\tiny GC}}}
\newcommand{\pgcstep}{\pstep_{\mbox{\tiny GC}}}
\newcommand{\cgcstep}{\rightarrow_{\mbox{\tiny CGC}}}

%\newcommand{\sunion}[2]{{#1} \stackrel{?}{\bigcup} {#2}}
%\newcommand{\spush}[2]{{#1} \stackrel{?}{\downarrow} {#2}}

% Other judgments
\newcommand{\fresh}{\ensuremath{\; \mathsf{fresh}}}
\newcommand{\starrow}[1]{\xrightarrow{#1}}
\newcommand{\alloc}[4]{#1; #2 \starrow{alloc} #3; #4}
\newcommand{\update}[4]{#1; #2; #3 \starrow{update} #4}
\newcommand{\lookup}[3]{#1; #2 \starrow{lookup} #3}
\newcommand{\newtask}[2]{#1 \starrow{new} #2}
\newcommand{\isdone}[3]{#1 \starrow{done} #2; #3}
\newcommand{\diffs}[1]{\mathit{diff}(#1)}
\newcommand{\initial}{\ensuremath{\;\mathsf{initial}}}
\newcommand{\htyped}[3]{\vdash_{#3} #1 : #2}

% Multilevel heap judgments
\newcommand{\heaptype}[3]{\left(#1, #2\right) : #3}
\newcommand{\allocg}[4]{\mathit{allocg}\left(#1, #2\right) = \left(#3, #4\right)}
\newcommand{\allocl}[4]{\mathit{allocl}\left(#1, #2\right) = \left(#3, #4\right)}
\newcommand{\promote}[6]{#1; #2; #3 \starrow{promote} #4; #5; #6}
\newcommand{\promotebrl}[3]{#1; #2; #3}
\newcommand{\promotebra}{\starrow{promote}}
\newcommand{\promotebrr}[3]{#1; #2; #3}
\newcommand{\pmap}{M}
\newcommand{\greachable}[1]{\mathsf{greachable}\left(#1\right)}

% Theorems
\newtheorem{thm}[theorem]{Theorem}
\newtheorem{lem}[theorem]{Lemma}
% \newtheorem{corollary}[theorem]{Corollary}
% \newtheorem{claim}{Claim}

%% Rule Array
\newenvironment{rulearray}
{
\newcommand{\newcol}{\qquad}
\newcommand{\newcolhalf}{\quad}
\newcommand{\newrow}{\\[4ex]}
\newcommand{\newrowhalf}{\\[2ex]}
\[
\begin{array}{c}
}
{
\end{array}
\]
\let\newcol\undefined
\let\newrow\undefined
}

% Author-specific todo notes
\newcommand{\ramtodo}[2][]
{\todo[color=magenta,author=Ram,size=\small,#1]{#2}}

\newcommand{\defn}[1]{\emph{\textbf{#1}}}
\newcommand{\mpl}{\textsf{MPL}}

\newcommand{\rulereftwo}[2]{rules~\rulename{#1} and \rulename{#2}}
\newcommand{\with}{\ensuremath{\mathbin;}}

\newcommand{\highlight}[1]{\colorbox{gray!20}{\ensuremath{#1}}}
\newcommand{\hred}[1]{\colorbox{red!10}{\ensuremath{#1}}}
\newcommand{\hblue}[1]{\colorbox{blue!10}{\ensuremath{#1}}}
\newcommand{\hgreen}[1]{\colorbox{green!10}{\ensuremath{#1}}}
\newcommand{\sizeof}[1]{\ensuremath{\lvert #1 \rvert}}
\newcommand{\costof}[1]{\ensuremath{\mathbf{cost} ({#1})}}
\newcommand{\cost}{\ensuremath{\mathbf{cost}}}
\newcommand{\oracle}{\ensuremath{\mathbf{oracle}}}

% inline "math highlight" to make it easier to read inline judgements
\definecolor{darkblue}{HTML}{0007C9}
\newcommand{\mh}[1]{{\ensuremath{\mbox{\ensuremath{#1}}}}}

%% variable context, location signature
% \newcommand{\ctxvar}{\Gamma}
% \newcommand{\ctxloc}{\Sigma}
\newcommand{\ctxemp}{\ensuremath{\cdot}}
\newcommand{\ctxext}[3]{\ensuremath{#1,#2\!:\!#3}} % extend context
\newcommand{\etyped}[4]{\ensuremath{{#1} \vdash_{#2} {#3} : {#4}}}
\newcommand{\memtyped}[3]{\ensuremath{{#1} \vdash {#2} : {#3}}}
\newcommand{\gtyped}[3]{\ensuremath{{#1} \vdash {#2} : {#3}}}
\newcommand{\httyped}[6]{\ensuremath{{#1} \with {#2} \with {#3} \vdash {#4}\!\cdot\!{#5} : {#6}}}
\newcommand{\ttyped}[5]{\ensuremath{{#1} \with {#2} \with {#3} \vdash {#4} : {#5}}}

\newcommand{\sttyped}[6]{\ensuremath{{\vdash_{#1} {#2} \with {#3} \with {#4} \with {#5} : {#6}}}}
\newcommand{\getyped}[6]{\ensuremath{{#1} \vdash_{#2, #3} {#4} \with {#5} : {#6}}}

%% types
\newcommand{\typnat}{\kw{nat}}
\newcommand{\typint}{\kw{int}}
\newcommand{\typbool}{\kw{bool}}
\newcommand{\typchar}{\kw{char}}
\newcommand{\typfloat}{\kw{float}}
\newcommand{\typprod}[2]{\ensuremath{{#1} \times {#2}}}
\newcommand{\typfun}[2]{\ensuremath{{#1}\!\rightarrow\!{#2}}}
\newcommand{\typref}[1]{\ensuremath{{#1}~\kw{ref}}}
\newcommand{\typfut}[1]{\ensuremath{{#1}~\kw{fut}}}
\newcommand{\futs}[1]{\mathsf{Fut}(#1)}
\newcommand{\futsmem}[2]{\mathsf{Fut}(#1, #2)}

% expression syntax
\newcommand{\enat}[1]{\ensuremath{#1}}
\newcommand{\efun}[3]{\ensuremath{\kw{fun}~{#1}~{#2}~\kw{is}~{#3}}}
\newcommand{\epair}[2]{\ensuremath{\langle {#1}, {#2} \rangle}}
\newcommand{\eapp}[2]{\ensuremath{{#1}~{#2}}}
\newcommand{\efst}[1]{\ensuremath{\kw{fst}~{#1}}}
\newcommand{\esnd}[1]{\ensuremath{\kw{snd}~{#1}}}
\newcommand{\eref}[1]{\ensuremath{\kw{ref}~{#1}}}
\newcommand{\ebang}[1]{\ensuremath{\mathop{!}#1}}
\newcommand{\eupd}[2]{\ensuremath{#1 \mathop{:=} #2}}
\newcommand{\elet}[3]{\kw{let}~{#1}={#2}~\kw{in}~{#3}}
\newcommand{\epar}[2]{\ensuremath{\langle {#1}\mathbin\|{#2} \rangle}}

\newcommand{\purelang}{{\sc $\lambda^{P}$}}
\newcommand{\reflang}{{\sc $\lambda^{U}$}}

% task syntax
% \newcommand{\tleaf}[2]{\ensuremath{{#1}\!\cdot\!{#2}}}
% \newcommand{\tpar}[4]{\ensuremath{\dblangle{{#1}\!\cdot\!{#2}\mathbin\|{#3}\!\cdot\!{#4}}}}
% \newcommand{\tpar}[4]{\ensuremath{\llparenthesis\,{#1}\!\cdot\!{#2}\mathbin\|{#3}\!\cdot\!{#4}\,\rrparenthesis}}
% \newcommand{\ttpar}[2]{\ensuremath{\llparenthesis\,{#1}\mathbin\|{#2}\,\rrparenthesis}}
% \newcommand{\tparg}[6]{\ensuremath{\llparenthesis\,{#1}\!\cdot\!{#2}\!\cdot\!{#3}\mathbin\|{#4}\!\cdot\!{#5}\!\cdot\!{#6}\,\rrparenthesis}}
% \newcommand{\tpar}[3]{\ensuremath{{#1}\!\cdot\!\llparenthesis\,{#2}\mathbin\|{#3}\,\rrparenthesis}}
% \newcommand{\taskhpe}[3]{\ensuremath{{#1}\!\cdot\!{#2}\!\cdot\!{#3}}}

% \newcommand{\mem}{\mu}
\newcommand{\mememp}{\emptyset}
\newcommand{\memext}[3]{\ensuremath{#1}[{#2} \!\hookrightarrow\! {#3}]}

\newcommand{\actarrow}{\blacktriangleright}
\newcommand{\pasarrow}{\vartriangleright}
\newcommand{\fmap}{\Delta}
\newcommand{\femp}{\emptyset}
\newcommand{\fmapactive}[3]{\ensuremath{#1} [{#2} \!\actarrow\! {#3}]}
\newcommand{\fmapjoined}[3]{\ensuremath{#1} [{#2} \!\pasarrow\! {#3}]}

\newcommand{\futctxt}{\Knownctxt}
\newcommand{\Futctxt}{\Knownctxt}
\newcommand{\ReadLocs}{\mathsf{R}}
\newcommand{\Knownctxt}{{K}}
\newcommand{\fut}[2]{\kw{fut}(#1; #2)}
\newcommand{\harpfut}[1]{\kw{fut}(#1)}
\newcommand{\futctxtemp}{\emptyset}
\newcommand{\te}[1]{\{#1\}}
\newcommand{\hemp}{\emptyset}
\newcommand{\hcat}{\cup}
\newcommand{\hext}[2]{{#1},{#2}}

\newcommand{\tack}{\oplus}
\newcommand{\plug}{\bowtie}

\newcommand{\omparam}{step length}
\newcommand{\actwrite}[2]{\textbf{U}{#1}\!\Leftarrow\!{#2}}
\newcommand{\actalloc}[2]{\textbf{A}{#1}\!\Leftarrow\!{#2}}
\newcommand{\actread}[2]{\textbf{R}{#1}\!\Rightarrow\!{#2}}
\newcommand{\actsync}[2]{\textbf{F}{#1}\!\Rightarrow\!{#2}}
\newcommand{\actnone}{\textbf{N}}

% Relations
\newcommand{\stepstar}{\longmapsto^*}
\newcommand{\tstepstar}{\tstep^*}
\newcommand{\drfstep}[2]{\xmapsto[{#2}]{{\,#1\,}}}
\newcommand{\drfstepstar}[1]{\xmapsto{{\,#1\,}}\joinrel\mathrel{^*}}

% computation graph
\newcommand{\gt}[2]{\ensuremath{\mathsf{GT}({#1},{#2})}}
\newcommand{\gemp}{\bullet}
\newcommand{\gseq}[2]{{#1}\oplus{#2}}
\newcommand{\gseqnamed}[3]{{#1}\oplus_{#2}{#3}}
\newcommand{\gseqa}[2]{\gseqnamed{#1}{a}{#2}}
\newcommand{\gseqb}[2]{\gseqnamed{#1}{b}{#2}}
\newcommand{\gspawn}[1]{\mathsf{spawn}\ {#1}}
\newcommand{\gsync}[1]{\mathsf{sync}\ {#1}}
\newcommand{\ghead}[1]{\mathsf{hd}(#1)}
\newcommand{\gtail}[1]{\mathsf{tl}(#1)}

\newcommand{\fcpar}[3]{\ensuremath{\gseq{#1}{(\gpar{#2}{#3})}}}

\newcommand{\gmerge}[2]{\bowtie_F ({#1}, {#2})}
\newcommand{\gmergerel}[3]{\bowtie_R ({#1}, {#2}) \downarrow {#3}}
\newcommand{\gcseq}[1]{\ensuremath{[#1]}}
\newcommand{\gcpar}[3]{\ensuremath{\gseq{#1}{(\gpar{#2}{#3})}}}
\newcommand{\gcparnamed}[4]{\ensuremath{\gseq{#1}{({#2}\otimes_{#3}{#4})}}}
\newcommand{\gcspawn}[4]{\ensuremath{\gseq{#1}{\gseq{#2}{(\gpar{#3}{#4})}}}}

\newcommand{\gpar}[2]{{#1}\otimes_{a}{#2}}
\newcommand{\gw}[1]{\ensuremath{\mathsf{W}({#1})}}
\newcommand{\ga}[1]{\ensuremath{\mathsf{A}({#1})}}
\newcommand{\greads}[1]{\ensuremath{\ReadLocs({#1})}}
\newcommand{\gaw}[1]{\ensuremath{\mathsf{AW}({#1})}}
\newcommand{\lw}[1]{\ensuremath{\mathsf{LW}({#1})}}
\newcommand{\alw}[1]{\ensuremath{\mathsf{A}({#1}) \cup \mathsf{LW}({#1})}}
\newcommand{\gabw}[1]{\ensuremath{\gaw{#1}}}

\newcommand{\saw}[1]{\ensuremath{\mathsf{SP}({#1})}}
\newcommand{\gf}[1]{\ensuremath{\overline{#1}}}

\newcommand{\extendsfj}[2]{\ensuremath{{#1}~\textsf{extends}~{#2}~\textsf{with f/j}}}
\newcommand{\extendswith}[3]{\ensuremath{{#1}~\textsf{extends}~{#2}~\textsf{with}~{#3}}}

\newcommand{\geok}[2]{{#1} \with {#2}~\textit{ok}}
\newcommand{\loc}[1]{{#1}~\textit{loc}}
\newcommand{\gleaf}[1]{{#1}~\textit{leaf}}
\newcommand{\gnode}[1]{{#1}~\textit{node}}

\newcommand{\drf}[2]{{#1} \vdash {#2}~{\textit{drf}}}
\newcommand{\drfb}[2]{{#1} \vdash {#2}~{\textit{wrf}}}
\newcommand{\drft}{\textit{drf}}
% \newcommand{\typed}[2]{{#1} \vdash }

% Theorems
\theoremstyle{plain}
\newtheorem{property}{Property}

% \theoremstyle{definition}
% \newtheorem{definition}{Definition}

%% Rules description
\newcommand{\flushLR}[3]{\hspace*{#3}\makebox[0em][l]{#1}\hspace*{\fill}\makebox[0em][r]{#2}\hspace*{#3}}
\newcommand{\rulesdesc}[2]{\textbf{#1}\hspace*{1em}{\fbox{#2}}}
\newcommand{\desc}[1]{\textbf{#1}}

% Syntax highlighting
\newdimen\zzlistingsize
\newdimen\zzlistingsizedefault
\zzlistingsizedefault=9pt
\newdimen\kwlistingsize
\kwlistingsize=9pt
\zzlistingsize=\zzlistingsizedefault
\gdef\lco{black}
%\newcommand{\keywordstyle}{\fontsize{0.9\zzlistingsize}{1.0\zzlistingsize}\bf}
%\newcommand{\keywordstyle}{\fontsize{\kwlistingsize}{1.1\kwlistingsize}\normalfont\bf\color{\lco}}
%\settowidth{\zzlstwidth}{{\Lstbasicstyle~}}

\begin{abstract}
  %% The field of compilers and program optimization, which is as nearly
%% old as computer science itself, has evolved largely synchronously with
%% architectural advances.
%% %
%% %% The development of deeply pipelined architectures with large memory
%% %% hierarchies have motivated optimization techniques for instruction
%% %% selection, scheduling, and register allocation.
%% %
%% For example, recent advances on parallel architectures such as multicores,
%% GPUs, and TPUs and have motivated the development of optimization
%% techniques for concurrent executing threads.
%
Recent advances in quantum architectures and computing have motivated
the development of new optimizing compilers for quantum programs or
circuits.
Even though steady progress has been made, existing quantum
optimization techniques remain asymptotically and practically
inefficient and are unable to offer guarantees on the quality of the
optimization.
Because many global quantum circuit optimization problems belong to
the complexity class QMA (the quantum analog of NP), it is not clear
whether quality and efficiency guarantees can both be achieved.

In this paper, we present optimization techniques for quantum programs
that can offer both efficiency and quality guarantees.
Rather than requiring global optimality, our approach relies on a form
of local optimality that requires each and every segment of the
circuit to be optimal.
We show that the local optimality notion can be attained by a
cut-and-meld circuit optimization algorithm.
%it into a hierarchy of pieces, optimizing each piece, and
%combining them together.
%
The idea behind the algorithm is to cut a circuit into subcircuits,
optimize each subcircuit independently by using a specified ``oracle''
optimizer, and meld the subcircuits by optimizing across the cuts
lazily as needed.
We specify the algorithm and prove that it ensures local optimality.
To prove efficiency, we show that, under some assumptions, the main
optimization phase of the algorithm requires a linear number of calls to
the oracle optimizer.
We implement and evaluate the local-optimality approach to circuit
optimization and compare with the state-of-the-art optimizers.
The empirical results show that our cut-and-meld algorithm can
outperform existing optimizers significantly, by more than an order of
magnitude on average, while also slightly improving optimization
quality.
These results show that local optimality can be a relatively strong
optimization criterion and can be attained efficiently.

\end{abstract}
\maketitle
\section{Introduction}
\if0
\paragraph{TODO: WEAVE: THREE QUESTIONS}
\begin{itemize}
    \item Does a locally optimal circuit exist?
    \item How does the quality of a locally optimal circuit compare with a global
    optimality?
    \item Is it possible to find a locally optimal circuit efficiently?
    \item Does our local optimization technique work well with different optimizers?
\end{itemize}

... our results are:
\begin{itemize}
    \item Definition of local optimality (semantics)
    \item Evidence that local optimality is, in practice, ``just as good
    as global optimality'' in terms of quality
    \item Algorithm for local optimality and proof of (nearly) linear time
    \item Empirical evaluation showing that our algorithm scales linearly
    with circuit size.
\end{itemize}
\fi

% quantum importance
Quantum computing holds the potential to solve problems in fields such
as chemistry simulation~\cite{Feynman82,Benioff80},
optimization~\cite{childs2017quantum, peruzzo2014variational},
cryptography~\cite{shor1994algorithms}, and machine
learning~\cite{biamonte2017quantum, schuld2015introduction}
that can be very challenging for classical computing techniques.
Key to realizing the advantage of quantum computing in these and
similar fields is achieving the scale of thousands of qubits and
millions of quantum operations (a.k.a., gates), often with high
fidelity (minimal error).
%
%high-fidelity quantum operations
~\cite{hoefler2023disentangling, gidney2021factor, alexeev2021quantum}.
Over the past decade, the potential of quantum computing and the
challenges of scaling it have motivated much work on both hardware and
software.
On the hardware front, quantum computers based on
superconducting circuits~\cite{kjaergaard2020superconducting},
trapped ions~\cite{monroe2021programmable, moses2023race}, and
Rydberg atom arrays~\cite{ebadi2021quantum, scholl2021quantum} have
advanced rapidly, scaling to hundreds of qubits and achieving
entanglement fidelity over $99\%$.
On the software front, a plethora of programming
languages,
optimizing compilers, and run-time environments have been proposed,
both in industry and in academia
(e.g.,~\cite{selinger-towards-2004,quipper-2013,prz-qwire-2017,Nam_2018,tket-2020,silq-2020,hietala2021verified,yc-tower-2022,twist-2022,quartz-2022,qiskit-2023,v+qunity-2023}).

%% % quantum importance
%% Quantum computers hold the potential to solve problems that are intractable for classical computers in fields such as chemistry simulation~\cite{Feynman82,Benioff80}, optimization~\cite{childs2017quantum, peruzzo2014variational}, cryptography~\cite{shor1994algorithms}, and machine learning~\cite{biamonte2017quantum, schuld2015introduction}.
%% %
%% However, these applications often require thousands of qubits and
%% millions of gates, often with high fidelity
%% %
%% %high-fidelity quantum operations
%% ~\cite{hoefler2023disentangling, gidney2021factor, alexeev2021quantum}.
%% % the next sentence seems to contradict
%% %, which are still beyond the capabilities of today's quantum computer systems.
%% %
%% Quantum devices based on superconducting circuits~\cite{kjaergaard2020superconducting}, trapped ions~\cite{monroe2021programmable, moses2023race}, and Rydberg atom arrays~\cite{ebadi2021quantum, scholl2021quantum} have advanced rapidly in recent years, scaling to hundreds of qubits and achieving entanglement fidelity over $99\%$.
%

%
Due to the limitations of modern quantum hardware and the need for
scaling the hardware to a larger number of gates, optimization of
quantum programs or circuits remain key to realizing the potential of
quantum computing.
The problem, therefore, has attracted significant research.
Starting with the fact that global optimization of circuits is QMA
hard and therefore unlikely to succeed, Nam et al
developed a set of heuristics for optimizing quantum programs or
circuits~\cite{Nam_2018}.
Their approach takes at least quadratic time in the number of gates
in the circuit, making it difficult to scale to larger circuits,
consisting for example hundreds of thousands of gates.
In followup work Hietala et al.~\cite{hietala2021verified} presented a
verified implementation of Nam et al.'s approach.
In more recent work Xu et al.~\cite{quartz-2022} presented techniques for
automatically discovering peephole optimizations (instead of human-generated heuristics) and applying them to optimize a circuit.
Xu et al.'s optimization algorithm, however, requires exponential time
in the number of the optimization rules and make no quality guarantees
due to pruning techniques used for controlling space and time
consumption.
In follow-up work Xu et al.~\cite{queso-2023} and Li et
al~\cite{li2024quarl} improve on Quartz's run-time.
All of these optimizers can take hours to optimize moderately large
circuits (\secref{eval})
and cannot make any quality guarantees.

Given this state of the art and the fact that global optimality is unlikely
to be efficiently attainable due to its QMA hardness~\cite{Nam_2018},
we ask:
\emph{\textbf{is it possible to offer a formal quality guarantee while also ensuring 
efficiency?}}

In this work, we answer this question affirmatively and thus bridge quality
and efficiency guarantees.
We first present a form of ``local optimality'' and its slightly
weaker form called ``segment optimality'', and present a rewriting
semantics for achieving local optimality.
For the rewriting semantics, we consider a reasonably broad set of
cost functions (as optimization goals) and
prove that saturating applications of local rewriting rules yield
local optimality.
To ensure generality, we formulate local optimality in an
``unopinionated'' fashion in the sense that we do not make any
assumptions about which optimizations may be performed by the local
rewrites.
Instead, we defer all optimization decisions to an
abstract \defn{oracle} optimizer that can be instantiated with an
available optimizer as desired.
Local optimality differs from global optimality in the sense that it
requires that each segment of the circuit, rather than the global
circuit, is optimal with respect to the chosen oracle.
We believe that this is a strong optimality guarantee, because it
requires optimality of each and every segment of the circuit.

Our rewriting semantics formalizes the notion of local optimality, but
it does not yield an efficient algorithm.
We present a local-optimization algorithm, called \algname{}
(Optimize-and-Compact) that takes a circuit and optimizes it in
rounds, each of which consists of an optimization and compaction
phase.
The optimization phase takes the circuit and outputs a segment-optimal
version of it, and the compaction phase compacts the circuit by
eliminating ``gaps'' left by the optimization, potentially enabling new
optimizations.
The \algname{} algorithm repeats the optimization and compaction
phases until convergence, where no more optimizations may be found.

To ensure efficiency, the optimization phase of \algname{} employs a
variant of the circuit cutting technique that was initially developed
for simulation of quantum circuits on classical
hardware~\cite{circuit-cutting-2020,tang-cutqc-2021,bravyi-future-qc-2022,k+cutting-2024}.
Specifically, our algorithm cuts the circuit hierarchically into
smaller subcircuits, optimizes each subcircuit independently, and
combines the optimized subcircuits into a locally optimal circuit.
To optimize small segments, the algorithm uses any chosen oracle and
does not make any restrictions on the optimizations that may be
performed by the oracle.
The approach can therefore be used in conjunction with many existing
optimizers that support different gate sets and cost functions.
By cutting the circuit into smaller circuits, the algorithm ensures
that most of the optimizations take place in the context of small
circuits, which then helps reduce the total optimization cost.
But optimizing subcircuits independently can miss crucial
optimizations.
We therefore propose a \defn{melding} technique to ``meld'' the
optimized subcircuits by optimizing over the cuts.
To ensure efficiency, our melding technique starts at the cut,
optimizes over the cut, and proceeds deeper into the circuit only as
needed.

The correctness and efficiency properties of our cut-and-meld
algorithm are far from obvious.
In particular, it may appear possible that 1) the algorithm misses
optimizations and 2) the cost of meld operations grows large.
We show that none of these are possible and establish that the
optimization algorithm guarantees segment optimality and accepts a
linear time cost bound in terms of the call to oracle.
For the efficiency bound, we use an ``output-sensitive'' analysis
technique that charges costs not only to input size but also to the
cost improvement, i.e., reduction in the cost (e.g., number of gates)
between the input and the output.
%
%% Taken together we are able to bound the time for the optimization
%% phase of the algorithm in terms of the size of the circuit and the
%% improvement in the cost function being targeted, which is linear in
%% many casse of interest (e.g., when optimizing count of all or some
%% gates).
%% %
Even though the algorithm can in principle take a linear number of
rounds, this appears unlikely, and we observe in practice that it
requires very few rounds (e.g., less than four on average).

%\ur{Update needed based on experiments.}

To evaluate the effectiveness of local optimality, we implement
the \algname{} algorithm
%, and
%instantiate it with two state-of-the-art
%optimizers as oracles: \voqc{} ~\cite{hietala2021verified}
%and~\feyntool{}~\cite{amy2019formal}.
%
and evaluate it by considering a variety of quantum circuits.
%
%in two different, broadly accepted gate sets.
%
Our experiments show that our \algname{} implementation improves
efficiency, by more than one order of magnitude (on average), and
closely matches or improves optimization quality.
These results show that local optimality is a reasonably strong
optimization criterion and our cut-and-meld algorithm can be a
efficient approach to optimizing circuits.
Because our approach is generic, and can be tooled to use existing
optimizers, it can be used to amplify the effectiveness of existing
optimizers to optimize large circuits.

Specific contributions of the paper include the following.
\begin{itemize}
    \item The formulation  of local optimality and its formal definition.

\item A rewriting semantics for local optimality and proofs that saturating rewrites yield locally optimal circuits.

    \item An algorithm \algname{} for optimizing quantum circuits locally.

\item Proof of correctness of \algname{}.
    
  \item Run-time complexity bounds and their proofs for  the \algname{} algorithm.
    
    \item Implementation and a comprehensive empirical evaluation of \algname{}, demonstrating
    the benefits of local optimality and giving experimental evidence for the practicality of the approach.
\end{itemize}

We note that due to space restrictions, we have omitted proofs of
 correctness and efficiency;
\iffull
we provide these proofs and additional 
experiments in the Appendix.
\else
we provide these proofs and additional 
experiments in the Appendix submitted as supplementary material.
\fi

\section{Background}

In this section,
we provide some quantum computing background that is relevant for the paper.

\paragraph{Quantum States, Gates, and Circuits}
The state of a quantum bit (or \emph{qubit}) is represented as a linear superposition,
$\ket{\psi} = \alpha\ket{0} + \beta\ket{1}$, of the single-qubit basis vectors $\ket{0} = [1\;\; 0]^\textnormal{T}$ and $\ket{1} = [0\;\; 1]^T$,
for $\alpha,\beta \in \mathbb{C}$ with normalization constraint $|\alpha|^2 + |\beta|^2 = 1$.
A valid transformation from one quantum state to another is described
as a $2\times 2$ complex unitary matrix, $U$,
where $U^\dagger U = I$. An $n$-qubit quantum state is a superposition of $2^n$ basis vectors,
$\ket{\psi} = \sum_{i\in\{0,1\}^n} \alpha_i \ket{i}$,
and its transformation is a $2^n\times 2^n$ unitary matrix.

A \emph{quantum circuit} is an ordered sequence of quantum logic gates selected from a predefined gate set.
Each \emph{quantum gate} represents a unitary matrix that transforms the state of one, two, or a few qubits.
%
% A set of gates is universal when it can be used to express all unitary transformations.
% %
% For example, some common \emph{universal gate sets} are
% $\{$\lstinline{H},\lstinline{RX},\lstinline{RY},\lstinline{RZ},\lstinline{CNOT}$\}$
% for noisy intermediate-scale quantum (NISQ) architectures
% and $\{$\lstinline{H},\lstinline{X},\lstinline{Y},\lstinline{Z},\lstinline{S},\lstinline{T},\lstinline{CNOT}$\}$
% for fault-tolerant quantum computing (FTQC) architectures.
%
Given a circuit $C$, the \emph{size} ($|C|$) is the total number of gates used,
while the \emph{width} ($n$) represents the number of qubits.
The \emph{depth} ($d$) is the number of circuit layers, wherein each qubit participates in at most one gate.
%\yr{Do we need a paragraph on QASM?}
%\jr{Add QASM + layer representation of circuits}

\paragraph{Circuit Representation.}
Quantum circuits can be represented with many data structures such as
graphs, matrices, text, and layer diagrams.
In this paper,
we represent circuits with a sequence of layers,
where each layer contains gates that may act at the same time step
on their respective qubits.
We use the layer representation to define and prove the circuit quality guaranteed
by our optimization algorithm.
In addition to the layer representation,
our implementation uses the QASM (quantum assembly language) representation.
The QASM is a standard format
which orders all gates of the circuit in a way that
respects the sequential dependencies between gates.
It is supported by almost all quantum computing
frameworks and enables our implementation to interact with off-the-shelf tools.

\paragraph{Quantum Circuit Synthesis and Optimization}
The goal of \emph{circuit synthesis} is to decompose the desired unitary transformation
into a sequence of basic gates that are physically realizable within the constraints of the underlying quantum hardware architecture.
Quantum circuits for the same unitary transformation can be represented in multiple ways,
and their efficiency can vary when executed on real quantum devices.
\emph{Circuit optimization} aims to take a given quantum circuit as input
and produce another quantum circuit that is
logically equivalent but requires fewer resources or shorter execution time,
such as a reduced number of gates or a reduced circuit depth.
Synthesizing and optimizing large circuits are known to be challenging due to their high dimensionality.
For example, as the number of qubits in a quantum circuit increases, the degree of freedom in the unitary transformation grows exponentially, leading to higher synthesis and optimization complexity. In particular, global optimization of quantum circuits is QMA-hard ~\cite{janzing2003identity}.
%
% Some analytical solutions exist, but they are typically restricted to small circuit width/depth ~\cite{rakyta2022approaching, giles2013remarks},
% or restricted to specific gate sets ~\cite{giles2013remarks}.
% %
% For large circuits, there is currently no general solution.
% %
% Researchers have thus proposed various heuristic-driven approaches~\cite{}.
%  heuristics.

% Given that practical algorithms that provide quantum advantage
% over their best-known classical counterparts
% frequently demand a large number of qubits or deep quantum circuits,
% the optimization of large/deep quantum circuits is crucial for their realization.

\section{Local Optimality}
\label{sec:lang}
%\label{sec:lang}
%\ur{TODO: rename ``shiftleft''  ``compaction''?}
%\ur{TODO: rename ``LOPT''  ``OPTIMIZATION''?}

In this section, we introduce \defn{local optimality} for quantum circuits
using a circuit language called \lang{}, which represents circuits as sequences of layers.
%
% \lang{} represents circuits as a sequence of layers.
%
We define local optimality based on three components:
(1) a base optimizer called the \defn{oracle},
(2) a \defn{cost} function that evaluates the circuit quality, and
(3) a \defn{segment size} $\Omega$,
which determines the scope of \emph{local optimizations}.
A segment refers to a contiguous sequence of layers.
Roughly speaking,
a circuit is locally optimal when it satisfies the following conditions:
\begin{enumerate}
    \item No local optimizations are possible, i.e., the oracle
    cannot optimize any segments of size $\Omega$.
    \item All circuit segments are as compact as possible with no unnecessary gaps.
\end{enumerate}

For a locally optimal circuit, the oracle cannot find more
optimizations unless it operates on segments larger than $\Omega$.
Because our definition is parametric in terms of the oracle,
we can define local optimality for any quantum gate set
by instantiating appropriate oracles.

We then develop a circuit rewriting semantics that produces
locally optimal circuits.
The semantics only uses the oracle on small circuit segments,
each containing at most $\Omega$ contiguous layers.
Using this semantics,
we prove that for a general class of cost functions,
\emph{any} circuit can be transformed into a locally optimal circuit.
This makes local optimality applicable to various cost functions
such as gate count, $\mathsf{T}$ count, and many others.

% %
% Our approach is also independent of the gate set, and can be instantiated for different
% gate sets by selecting an appropriate oracle.

% Local optimality is a practical goal of optimization
% because many optimizers struggle to optimize
% large circuits.
% This is useful becauseand are only effective when the circuits are small.
% %
% Because optimizing the whole circuit in one call to the optimizer can be
% infeasible, we use local optimality as the goal of optimization.

%
% This parameter, $\Omega$, controls how ``local'' the local optimality
%

%

% specifically by upper-bounding the number of circuit layers that
% are given to the oracle.
% %
% This bound, denoted $\Omega$, is another parameter of the system.
% %
% Altogether, our approach therefore has three parameters:
% an oracle $\oracle{}$, a \costof{} function, and a

% and our optimality results are taken relative
% to the oracle.
% %
% That is,
% approach can be applied to optimize across a variety
% %
% Local optimizations must adhere to a cost function
% %
% The rewrite system is quite general because it is parameterized by
% takes as parameters a base optimizer, a cost function, and a \omparam{} $\Omega$,
% and define rules for local optimization.
% %
% Because its parametric, the rewrite system can be used with a variety of optimizers,
% gate sets, and cost functions.

\subsection{Circuit Syntax and Semantics}

We present our circuit language called \lang{} (Layered Quantum Representation)
which represents a quantum circuit as a sequence of layers.
Figure~\ref{fig:lang:syn} shows the abstract syntax of the language.
We let the variable $q$ denote a qubit, and $G$ denote a gate.
For simplicity, we consider only unary gates $g(q)$ and binary gates
$g(q_1, q_2)$, where $g$ is a gate name in the desired gate set.
These definitions can be easily extended to support gates of any arity.

A \lang{} circuit $C$ consists of a sequence of
layers $\langle L_0, \dots, L_{n - 1} \rangle $,
where each layer $L_i$ is a set of gates
that are applied to qubits in parallel.
The circuit is \defn{well formed} if
the gates of every layer act on disjoint qubits, i.e.,
no layer can apply multiple gates to the same qubit.
As a shorthand, we write $\compat{G_1}{G_2}$ to denote that gates
$G_1$ and $G_2$ act on disjoint qubits,
i.e., $\qubits{G_1} \cap \qubits{G_2} = \emptyset{}$.
We similarly write $\compat{L_1}{L_2}$ for the same condition on layers.
Note that we implicitly assume well-formedness throughout the section
because it is preserved by all our rewriting rules.

We define the \defn{length} of a \lang{} circuit as the number of layers,
and the \defn{size} of a circuit $C$, denoted $|C|$, as the total number of gates.
A \defn{segment} is a contiguous subsequence of layers of the circuit, and
a \defn{k-segment} is a segment of length $k$.
%
% The size and length of a segment are defined similarly to circuits.
We use the Python-style notation $C[i : j]$ to represent a segment containing
layers $\langle L_i, \dots, L_{j - 1} \rangle $ from circuit $C$.
In the case of overflow (where either $i < 0$ or $j > \mathsf{length}(C)$),
we define $C[i : j] = C[\max(0, i) : \min(j, \mathsf{length}(C))]$.

Two circuits $C$ and $C'$ can be concatenated together as $C; C'$,
creating a circuit containing the layers of $C$ followed by layers of $C'$.
Formally,
if $C = \langle L_0 \dots L_{n-1} \rangle$ and $C' = \langle L'_0 \dots L'_{m-1} \rangle$,
then $C;C' = \langle  L_0 \dots L_{n-1},L'_0 \dots L'_{m-1} \rangle$.

%
% Gate names can come from any standard Quantum gate set, and must
% be supported by the chosen oracle optimizer.
% (e.g., Hadamard gate, NOT gate, controlled-NOT gate, T gate, etc.)

% %
% The figure uses variable $q$ (and variants) to range over a set of
% qubits, indexed by natural numbers $1 \ldots n$,
% %
% and variable $G$ (and variants) to range over a set of gates.
% %
% For simplicity we only consider unary and binary gates operating on one
% and two qubits respectively.
% %

\begin{figure*}
{
% \small
% \fbox{
\begin{minipage}{0.45\textwidth}
\[
\begin{array}{llcl}

% \textit{Variables} & x, f & &
% \\
\textit{Qubit} & q & &
\\
\textit{Gate Name} & g & &
\\
\textit{Gate} & G & \bnfdef &
g\ (q) \bnfalt
g\ (q_1, q_2)
\\
\textit {Layer} & L & \bnfdef & 
% Qubits \rightharpoonup Gates
% Set(Gate)
\left\{G_0, G_1, \dots, G_{t-1}\right\}
\\
\textit {Circuit} & C & \bnfdef & \langle L_0, L_1, \ldots L_{n-1}\rangle
\end{array}
\]
\end{minipage}
% }
\hfill{}
\begin{minipage}{0.45\textwidth}
% \fbox{\begin{minipage}{\pagewidth}
\begin{align*}
\qubits{G} &\defeq \begin{cases}
    \{q\}, &G = g(q) \\
    \{q_1, q_2\}, &G = g(q_1, q_2)
\end{cases} \\
\qubits{L} &\defeq \bigcup_{G \in L} \qubits{G}
\end{align*}
% \end{minipage}}

% \fbox{\begin{minipage}{\pagewidth}
\begin{align*}
&\compat{G_1}{G_2} \Leftrightarrow \qubits{G_1} \cap \qubits{G_2} = \emptyset{} \\
&\compat{L_1}{L_2} \Leftrightarrow \qubits{L_1} \cap \qubits{L_2} = \emptyset{} \\
% &\quad\text{(and similarly $\compat{L_1}{L_2}$ for compatible layers)} \\
&C~\text{well-formed} \Leftrightarrow \\
&\quad \forall L \in C.~\forall G_1, G_2 \in L.~
% \\
% &\quad\quad \qubits{G_1} \cap \qubits{G_2} = \emptyset
\compat{G_1}{G_2}
\end{align*}
\end{minipage}
% \end{minipage}}
%
%
% \end{minipage}
\caption{Syntax of \lang{} and well-formed circuits.}
\label{fig:lang:syn}
}

\end{figure*}

% %
% Modern quantum optimizers work by taking typically small subcircuits of the input circuit and rewriting them to improve quality.
% %
% For this approach to work well, it is important for a layered circuit to be as compact as possible so as to maximize each subcircuit, which can be optimized.
%
%

% \sr{Paragraph above seems out of place here. Where should we motivate compactness?}

% \subsection{Well Typed Programs}

% We say that a program $P$ is well-typed if it consists of well-typed
% layers.
% %
% We say that a layer is well typed if each qubit is used (applied to a
% gate) at most once.
% %
% Well typedness enforces the no-cloning theorem of quantum mechanics,
% which prevents copying of an arbitrary quantum state.

\subsection{Local Optimality}

\begin{figure}
\small
\[
\begin{array}{c}

\infer[]
{
  \forall i.~
  \forall G \in L_{i}.~
  % \qubits{G} \cap \qubits{L_{i-1}} \neq \emptyset{}
  \notcompat{L_{i-1}}{\{G\}}
}
{
  \compressed{\langle L_0, \ldots, L_{n-1} \rangle}
}

% \\[3ex]

% \infer[]
% {
%   \forall G \in L_2.~\qubits G \cap \qubits {L_1} \neq \emptyset{}
% }
% {
%   \compressed{\langle L_1, L_2 \rangle}
% }

\qquad 

\infer[]
{
  \forall ij.~ i \leq j \leq i+\Omega \Rightarrow \costof{\oracle{}(C[i:j])} = \costof{C[i:j]}
}
{
  \windowopt\Omega C
}

\\[3ex]

\infer[]
{
  \compressed{C}
  \\
  % \forall i. \forall j.~ i \leq j \leq i+\Omega \Rightarrow \windowopt\Omega{C[i : j]}
  \windowopt\Omega C
}
{
  \locallyopt\Omega{C}
}

\end{array}
\]
% \addtolength{\belowcaptionskip}{-18pt}

\caption{Definition of local optimality, parameterized by a $\cost$ function, an \oracle{} optimizer,
and a segment length $\Omega$.}
\label{fig:lopt-defn}
\end{figure}

% \todo{introduction to the subsection}
We introduce two optimality properties on circuits
written in our language \lang{}.
These properties are defined on the following parameters.
First, we assume an abstract \cost{} function over circuits,
where a smaller cost means a better quality circuit.
Second, we introduce a parameter $\Omega$,
which represents the maximum segment length that can be considered for optimization.
In this context,
the optimizations are \emph{local}
because each optimization can only optimize a circuit segment of length $\Omega$.
Third,
we assume an $\oracle$ optimizer that takes a circuit of length $\Omega$ and
produces an equivalent circuit, optimized w.r.t. the cost function.
We assume that the oracle and the cost function are compatible,
meaning that the oracle can only decrease the cost:
\[ \forall C.~\costof{\oracle{}(C)} \leq \costof{C}. \]

\myparagraph{Segment optimal circuits.}
A circuit is segment-optimal
if each and every $\Omega$-segment of the circuit is optimal
for the given \oracle{} and \cost{} function.
\figref{lopt-defn} defines \wopttext{} circuits
as the judgment $\windowopt\Omega C$.
The judgment checks that any segment $C[i:j]$ whose length
is smaller than $\Omega$ (i.e., $i \leq j \leq i+\Omega$)
can not be further optimized by the oracle.
Thus,
calling the oracle on any such
segment does not improve the cost function.
%
% The optimality guarantee is slightly weaker because it does not require
% the circuit to be \emph{compact}, which we define next.

\myparagraph{Compact circuits.}
A circuit is \emph{compact} if every gate is in the left-most possible
position, i.e., in the earliest layer possible.
% \footnote{It would also be possible to define compactness with right-most}
%
\figref{lopt-defn} formalizes this with the judgement $\compressed C$.
The judgment checks every gate $G$ and ensures that
at least one qubit used by $G$ is also used by the previous layer
(i.e., $G \in L_i$ and $\notcompat{L_{i-1}}{\{G\}}$).
%
% Therefore, $G$ cannot be moved left.
%
If every gate satisfies this condition, the circuit is compact.

Ensuring that a layered circuit is as compact as possible
is important because compact circuits are more amenable to local optimizations.
For example, consider two $H$ gates on the same qubit, one of them in layer $0$
and the other in layer $3$, with no gates in between.
% For example, consider a circuit with an $H$ gate on layer $0$ at qubit $0$ and
% another $H$ gate on layer $5$ at qubit $0$, with no gates in between.
%
Suppose we have an optimization that cancels two $H$ gates on the same qubit,
when they are on adjacent layers.
This optimization would not apply to
our circuit because the two $H$ gates are not adjacent.
However, if the circuit did not have such unnecessary ``gaps'',
we could apply the optimization and eliminate the two $H$ gates from our circuit.
Compaction ensures that such optimizations are not missed.

\myparagraph{Locally Optimal Circuits.}
A circuit is {locally optimal} if it is both
compact and \wopttext.
This means that each and
every $\Omega$-segment of the circuit is optimal
for the given \oracle{} and \cost{} function
and the circuit as a whole is compact,
ensuring that no more local optimizations are possible.
\figref{lopt-defn} defines locally optimal circuits,
as the judgment $\locallyopt\Omega C$,
%

% ====================================================================================================
% ====================================================================================================
% ====================================================================================================

\subsection{Circuit Rewriting for Local Optimality}

In this section,
we present a rewriting semantics for producing locally optimal circuits.
The rewriting semantics performs local optimizations,
each of which rewrites an $\Omega$-segment,
and compacts the circuit as needed.
%
% The key result of this section is a
% theorem that states that under some assumptions
% a circuit is either locally optimal or can be optimized to improve along the cost function.
Given a $\cost{}$ function, a \oracle{}, and a segment length $\Omega$,
we define the rewriting semantics as a relation $\localstep\Omega C {C'}$
which rewrites the circuit $C$ to circuit $C'$.
%
% For improved readability, we treat the cost, $\Omega$, and the oracle parameters as implicit and write $\localstep\Omega{C}{C'}$.
%
\figref{circ-rewrite} shows the rewriting rules
$\rulename{Lopt}$ and $\rulename{ShiftLeft}$
for local optimization and compaction, respectively.

\begin{figure}
% \[
% \plug(a, \tree_1, \tree_2) =
%      \begin{cases}
%       \gseqa {\trace_1} {\tree_2}  & \text{\emph{if}}\ \tree_1 = \gcseq{\trace_1}\\
%        \gseqb{\trace_1}{\plug (a, \tree'_1, \tree_2)}  & \text{\emph{if}}\ \tree_1 = \gseqb{\trace_1}{\tree'_1} \\
%        \gcparnamed{\trace_1}{\tree'_1}{b}{\plug (a, \tree''_1, \tree_2)}  & \text{\emph{if}}\
%         \tree_1 = \gcparnamed{\trace_1}{\tree'_1}{b}{\tree''_1} \\
%      \end{cases}
% \]

% $
%   \begin{array}{c c c}
%      \plug(a, \gcseq{\trace_1}, \tree_2) & =  & \gseqa {\trace_1} {\tree_2} \\
%      \plug(a, \ \gseqb{\trace_1}{\tree'_1}, \tree_2) & =  & \gseqb{\trace_1}{\plug (a, \tree'_1, \tree_2)} \\
%      \plug(a, \gcparnamed{\trace_1}{\tree'_1}{b}{\tree''_1}, \tree_2) & =  & \gcparnamed{\trace_1}{\tree'_1}{b}{\plug (a, \tree''_1, \tree_2)} \\
%   \end{array}
% $

\[
\begin{array}{c}
% \small
\infer[Lopt]
{
  \mathsf{length}(C) \leq \Omega
  \quad
  C' = \oracle{}(C)
  \quad
  \costof{C'} < \costof{C}
}
{
  \localstep\Omega{P; C; S}{P; C'; S}
}
\\[3ex]
\infer[ShiftLeft]
{
  G \in L_2 \\
  % \forall q \in \qubits G.~ q \not\in \qubits{L_1} \\
  % \qubits{G} \cap \qubits{L_1} = \emptyset{} \\
  \compat{L_1}{\{G\}} \\
  L'_1 =  L_1 \cup \{G\} \\
  L'_2 = L_2 \setminus \{ G \} 
}
{
  \localstep\Omega{P; \langle L_1, L_2 \rangle; S}{P; \langle L'_1, L'_2 \rangle; S}
}

\end{array}
\]
% \addtolength{\belowcaptionskip}{-18pt}

\caption{Local optimization rewrite rules.}
\label{fig:circ-rewrite}
\end{figure}

% \input{fig/fig-rewrite-converge}

% At a high level,
% the rule \rulename{Lopt} defines one local optimization using the oracle optimizer $\oracle$
% on some circuit segment $C$ of length upto $\Omega$.
% %
% The second rule \rulename{ShiftLeft} \emph{compresses} the circuit to enable more optimizations.
% %
% We discuss these rules in detail below.

% %
% We represent this formally with the local optimization relation ${C}\longmapsto_{\cost{}, \Omega, \oracle}{C'}$,
% which holds if
% 1) $C'$ is a better quality circuit than $C$, i.e., $\costof{C'} < \costof{C}$,
% 2) $C$ and $C'$ are syntactically equivalent except for a segment of length $\Omega$
% which is rewritten according to the oracle $\mathcal{O}$,
% and
% 3) $C$ and $C'$ are semantically equivalent.
% %
% For ease of denoting the optimization relation,
% we leave the cost function $\cost{}$ and the oracle $\oracle$ implicit,
% and instead denote it as $\localstep\Omega{C}{C'}$.
% \input{fig/fig-rewrite}
% We specify the oracle and the cost function explicitly when needed.

\myparagraph{Local Optimization Rule.}
The rule $\rulename{Lopt}$ (\figref{circ-rewrite})
performs one local optimization on a circuit.
It takes a segment $C$ of length $\Omega$,
feeds it to the \oracle{},
and retrieves the output segment $C'$.
If the cost of $C'$ is lower than the cost of $C$,
then the rule replaces the segment $C$ with $C'$.
The rule does not modify the
remaining parts $P$ (prefix) and $S$ (suffix) of the circuit.
%

% The rule \rulename{Lopt} improves the circuit by
% performing local optimizations.
% %
% As the rule optimizes different segments of the circuit,
% it can add and remove gates and layers, which can create gaps in
% the layers of the circuit.
% %
% These gaps can stretch the circuit segments to be longer than necessary,
% making them ineligible for further optimization.

% For example,
% consider a circuit segment that can be represented in $\Omega$ layers when
% there are no gaps.
% %
% Without gaps,
% since its length is $\Omega$,
% we could apply the rule \rulename{Lopt} and optimize it using the oracle.
% %
% However, if there are gaps between the gates of the circuit segment,
% it would instead use more than $\Omega$ layers,
% making it ineligible for further optimization.
% %
% Thus, local optimizations may be missed
% if the circuit segments have gaps and are longer than necessary.

\myparagraph{Compaction Rule.}
To compact the circuit, we have the rule \rulename{ShiftLeft} (\figref{circ-rewrite}).
In one step,
the rule shifts a gate to the ``left''
by removing it from its current layer and adding it to the preceding layer.
A circuit is compact whenever all gates have been fully shifted left.
Formally, the \rulename{ShiftLeft} rule
considers consecutive layers $L_1$ and $L_2$ and moves a gate
$G$ from layer $L_2$ into $L_1$,
creating new layers $L_1' = L_1 \cup \{G\}$ and $L_2' = L_2 \setminus \{G\}.$
To maintain well-formedness,
the rule checks that no gate in the previous layer operates on the same qubits
($\compat{L_1}{\{G\}}$).

\myparagraph{Example.}
\begin{figure}[t]
    \centering
    \includegraphics[width=\columnwidth]{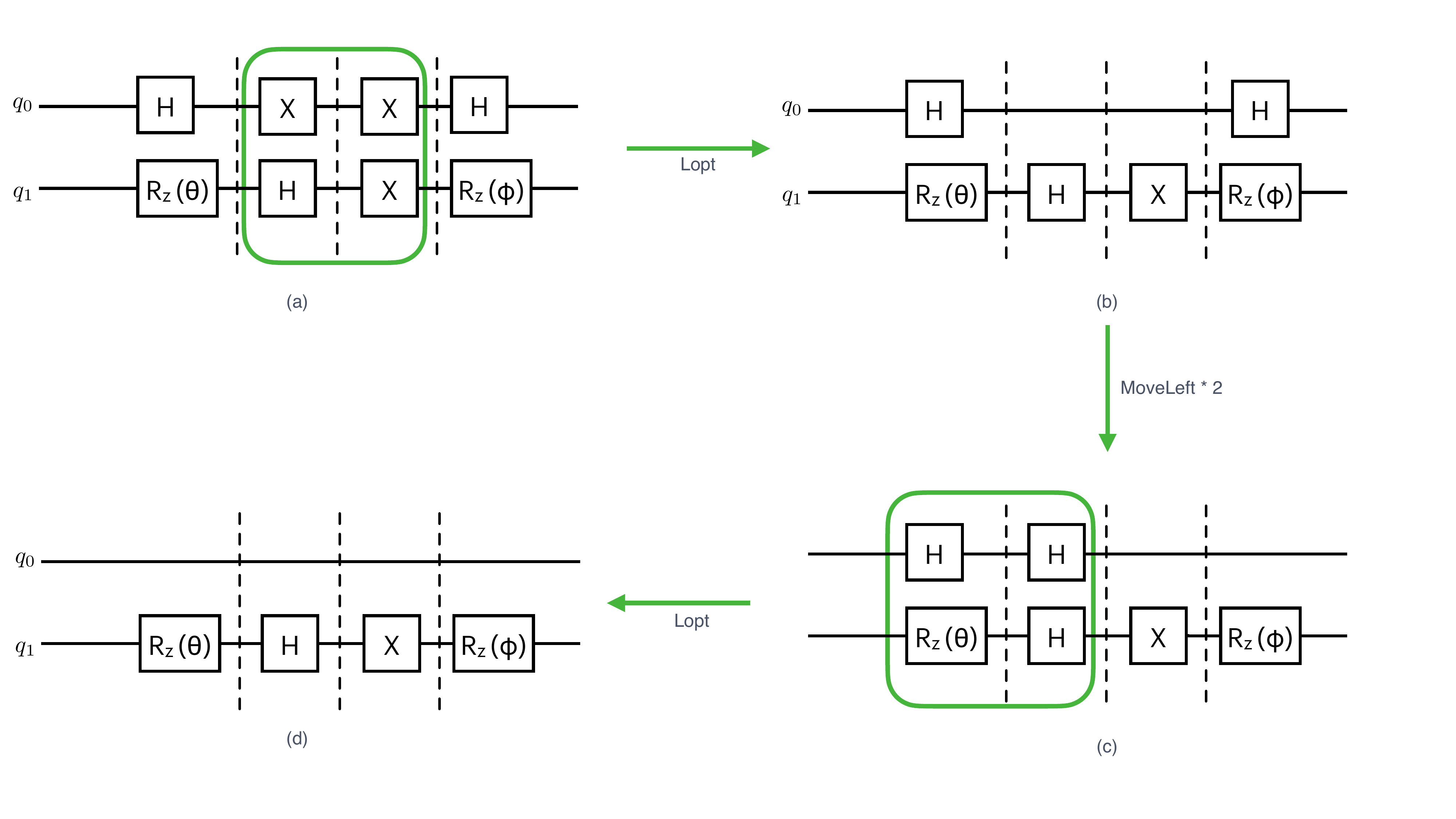}
    \vspace{-0.25in}
    \caption{The figure illustrates how our rewriting semantics optimizes circuits for $\Omega = 2$.
        The figure implicitly assumes an oracle that removes any two consecutive $H$ and $X$ gates.
        At each step, our semantics either selects a segment of size $2$ (denoted by green boxes)
        and performs an optimization,
        or picks a gate and shifts it left.
    }
    \label{fig:compression_unlocks_optimizations}
\end{figure}

\figref{compression_unlocks_optimizations} illustrates how
our rewriting rules can optimize a circuit with $\Omega = 2$.
The dotted lines in the circuit separate the four layers.
The optimizations in the figure implicitly use an oracle that
performs the following actions:
it removes two consecutive $H$ gates on the same qubit because they cancel each other;
similarly, it removes two consecutive $X$ gates on the same qubit.

In the figure,
moving from (a) to (b),
the rule \rulename{Lopt} optimizes the $2-$segment enclosed in the green box and
removes two consecutive X gates on qubit $q_0$.
Then,
going from (b) to (c),
we apply the rule \rulename{ShiftLeft} twice,
moving the gate $H$ gates together,
and compacting the circuit.
Because of this compaction,
the rule \rulename{Lopt} can step from (c) to (d),
canceling the two consecutive $H$ gates.

\subsection{Correspondence between local optimality and local rewrites}

It is easy to show that our definition of local optimality is consistent with the
rewriting semantics, in the sense that (1) locally optimal circuits cannot be further
rewritten (\lemref{converged}), and (2) if a circuit is not locally optimal, then it
always can be rewritten (\lemref{can-step}).
%
% \begin{lemma}
% \label{lem:well-formed}
% For every well-formed circuit $C$, if $\localstep\Omega C {C'}$ then $C'$ is well-formed.
% \end{lemma}

\begin{lemma}
\label{lem:converged}
For every circuit $C$ where $\locallyopt\Omega{C}$, there is no $C'$ such that $\localstep\Omega C {C'}$.
\end{lemma}

\begin{lemma}
\label{lem:can-step}
For every circuit $C$, either $\locallyopt\Omega C$ or $\localstep\Omega C {C'}$.
\end{lemma}

\subsection{Termination}
Given any circuit, we intend to use the rewriting semantics to produce an equivalent
locally optimal circuit.
However, this approach only succeeds if the rewriting semantics terminates, i.e., if
eventually no further rewrites are possible.
(Note that when our semantics terminates, \lemref{can-step} guarantees that the
circuit is locally optimal.)
Below are two examples where termination is not guaranteed.
\begin{enumerate}[leftmargin = *]
\item \textbf{Infinite Decrease in Cost}:
The cost function could be such that it decreases infinitely.
For instance,
consider the (contrived) cost function $\costof{C} = \prod_{\textsf{Rz}(\theta) \in C} \theta$,
which takes the product of the angles of gates $\textsf{Rz}(\theta)$ performing Z-axis rotations.
Suppose we have a circuit containing a single gate $\textsf{Rz}(1)$, setting $\theta = 1$
and its cost is $1$.
We can replace the $\textsf{Rz}(1)$ gate by  two consecutive $\textsf{Rz}(1/2)$ gates,
which achieve the same rotation,
but their combined cost is $1/4$.
This process can repeat infinitely,
halving the angle at each step, and the cost function would keep decreasing.
For this circuit and cost function, our rewriting semantics does not terminate.

\item \textbf{Cycles Due to Cost Function}:
A local optimization on one segment can increase the cost of a nearby segment,
creating an infinite loop of local optimizations.
To demonstrate this, consider the badly behaved cost function:
$\costof{C} = 0$ if $|C|$ is even, and
otherwise $\costof{C} = 1$.
Now, ``optimizing'' one segment (to toggle its number of gates from odd to even)
can cause a nearby overlapping segment
to toggle from even to odd, potentially repeating forever.
\end{enumerate}

% %
% We first discuss those cost functions for which the rewriting semantics can
% not terminate, because a locally optimal circuit does not even exist.
% %
% Then we discuss a class of cost functions for which our semantics terminates,
% producing locally optimal circuits.
% %
% % From \lemref{can-step}, we could attempt to obtain a locally optimal circuit by applying rewrite
% % rules until no more rewrites are possible.
% % %
% % However, for arbitrary cost functions, it is not guaranteed that these rewrites will ``terminate'' (i.e.,
% % converge upon a locally optimal circuit).

% We next consider whether or not local rewrites can be used to produce locally optimal
% circuits.
% %
% Our goal is to iteratively apply the rules \rulename{Lopt} and \rulename{ShiftLeft}
% until no more rules can be applied, at which point the circuit will be locally
% optimal.
% %
% In this section, we show that this approach is viable, but only if we place reasonable
% restrictions on the $\cost{}$ function.

\myparagraph{Ensuring termination.}
We can prove that our semantics terminates under certain conditions on the cost function.
First, we require $\costof{C} \in \mathbb{N}$; this guarantees that the cost cannot decrease infinitely.
Second, we require that the cost function is \defn{additive} according to the following definition.
\begin{definition}
A function $\cost{}: \textit{Circuit} \to \mathbb{N}$ is \defn{additive} iff both of the following conditions hold:
\begin{enumerate}
    \item $\costof{C_1 ; C_2} = \costof{C_1} + \costof{C_2}$ for all circuits $C_1$ and $C_2$.
    \item $\costof{\langle L_1 \cup L_2 \rangle} = \costof{\langle L_1 \rangle} + \costof{\langle L_2 \rangle}$ for all layers $L_1$ and $L_2$ such that $\compat{L_1}{L_2}$.
    % $\qubits{L_1} \cap \qubits{L_2} = \emptyset{}$.
\end{enumerate}
\end{definition}

Note that all cost functions that take a linear combination of counts
of each gate in the circuit are additive.
These include metrics
such as gate count (number of gates), $\mathsf{T}$ count (number of $\mathsf{T}$ gates), $\mathsf{CNOT}$ count (number of $\mathsf{CNOT}$ gates),
and two-qubit count (number of two-qubit gates).

For additive cost functions, we prove (\thmref{convergence}) that the rewriting semantics always
terminates.

\begin{theorem}[Termination]
\label{thm:convergence}
Let $\cost{}: \textit{Circuit} \to \mathbb{N}$ be an additive cost function.
For any initial circuit $C_0$,
there does not exist an infinite sequence of circuits such that
$$C_0 \overset{\Omega}\optstep C_1 \overset{\Omega}\optstep C_2 \overset{\Omega}\optstep \cdots$$
where each $\localstep{\Omega}{C_i}{C_{i+ 1}}$ represents a step of our rewriting semantics.
\end{theorem}

To prove \thmref{convergence}, we define a ``potential'' function
$\Phi: \textit{Circuit} \to \mathbb{N} \times \mathbb{N}$
and show that each step of our rewriting semantics decreases the potential function,
with the ordering $\Phi(C') < \Phi(C)$ defined lexicographically over $\mathbb{N} \times \mathbb{N}$.
A suitable definition for $\Phi$ is as follows.
\[
\begin{array}{c}
  \Phi(C) \defeq \left(\costof{C}, \mathsf{IndexSum}(C)\right)
\end{array}
\quad\text{where}\quad
\begin{array}{c}
  \mathsf{IndexSum}(C) \defeq \sum_{i} i |L_i| \\[1ex]
  C = \langle L_0, L_1, \ldots \rangle
\end{array}\]
This potential function has two components:
(1) the cost of the circuit, which decreases as optimizations are performed, and
(2) an ``index sum'', which decreases as gates are shifted left.
Together these components guarantee that the potential function decreases at each step, 
ensuring termination.
We provide the full proof in
\iffull
\appref{lem-pot-dec}.
\else
the Appendix.
\fi

\section{Local Optimization Algorithm} \label{sec:algorithm}
% This is the full version
%\input{algorithm}

In \secref{lang}, we presented a rewriting semantics consisting of just
two rules (corresponding to optimization of a segment and compaction)
and
proved that any saturating rewrite that applies the two rules to
exhaustion yields a locally optimal circuit.
This result immediately suggests an algorithm: simply apply the
rewriting rules until they no longer may be applied, breaking ties
between the two rules arbitrarily.
Even though it might seem desirable due to its simplicity, such an
algorithm is not efficient, because searching for a segment to
optimize requires linear time in the size of the circuit (both in
worst and the average case), yielding a quadratic bound for
optimization.
For improved efficiency, it is crucial to reduce the search time
needed to find a segment that would benefit from optimization.

Our algorithm, called \algname{}, controls search time by using a
circuit cutting-and-melding technique.
The algorithm cuts the circuit hierarchically into smaller
subcircuits, optimizes each subcircuit independently.
The hierarchical cutting naturally reduces the search time for the
optimizations by ensuring that most of the optimizations take place in
the context of small circuits.
Because the algorithm optimizes each subcircuit independently, it can
miss crucial optimizations.
To compensate for this, the algorithm melds the optimized subcircuits
and optimizes further the melded subcircuits starting with
the \defn{seam}, or the boundary between the two subcircuits.
The meld operation guarantees local optimality and does so efficiently
by first optimizing the seam and further optimizing into each
subcircuit only if necessary.
By melding locally optimal subcircuits, the algorithm can guarantee
that the subcircuits or any of their ``untouched'' portions (what it
means to be ``untouched'' is relatively complex) remain optimal.
We make this intuitive explanation precise by proving that the
algorithm yields a locally optimal circuit.
We note that circuit cutting techniques have been studied for the
purposes of simulating quantum circuits on classical
hardware~\cite{circuit-cutting-2020,tang-cutqc-2021,bravyi-future-qc-2022}.
We are not aware of prior work on circuit melding techniques that can
lazily optimize across circuit cuts.

\subsection{The algorithm}
\label{sec:opt-phase}

\begin{figure}
\centering
\input{fig/fig-lopt-code}
\caption{
Algorithm \algname{} produces locally optimal circuits with
respect to a given $\oracle{}$, $\cost{}$, and segment length $\Omega$.
To achieve local optimality,
\algname{} only uses the oracle on small segments of length $2\Omega$.
The algorithm repeatedly optimizes and compacts the circuit
until convergence.
The function \textsf{segopt}$()$ implements our optimization algorithm
and uses $\mathsf{meld}()$ to efficiently produce \wopttext{} circuits.
}
% \vspace{-1.5in}
\label{fig:lopt-code}
\end{figure}

\figref{lopt-code} shows the pseudocode for our algorithm.
The algorithm (\lstinline{OAC}) organizes the computation into rounds,
where each round corresponds to a recursive invocation
of \lstinline{OAC}.
A round consists of a compaction phase (via the function \textsf{compact})
and
a segment-optimization phase, via the function \textsf{segopt}.
The rounds repeat until convergence, i.e., until no more optimization
is possible (at which point the final circuit is guaranteed to
be \emph{locally optimal}).
As the terminology suggests, the segment optimization phase always
yields a segment-optimal circuit, where each and every segment is
optimal (as defined by our rewriting semantics).
Compaction rounds ensure that the algorithm does not miss optimization
opportunities that arise due to compaction.
%
%% Our experiments show that the number of rounds is quite small, even
%% for large circuits, and the vast majority of all optimizations are
%% performed in the first round.
%% %
%% We specifically observe that the number of rounds does not exceed $3$ in the majority of cases, and does not exceed $11$
%% across all of our benchmarks.
%% %
We also present a relaxed version of our algorithm that stops early
when a user-specified \defn{convergence threshold}
$0 \leq \epsilon \leq 1$, is reached.

\myparagraph{Function \textsf{segopt}.}
The function \textsf{segopt} takes a circuit $C$ and produces
a segment optimal output.
To achieve this, it uses a divide-and-conquer strategy to cut the
circuit hierarchically into smaller and smaller circuits:
it splits the circuit into the
subcircuits $C_1$ and $C_2$, optimizes each recursively, and then
calls \textsf{meld} on the resulting circuits to join them back
together without losing segment optimality.
This recursive splitting continues until the circuit has been partitioned into
sufficiently small segments,
specifically where each piece is at most $2\Omega$ in length.
For such small segments, the function directly uses the oracle
and obtains optimal segments.

\myparagraph{Function \textsf{meld}.}
\figref{lopt-code} (right) presents the pseudocode of the meld function.
The function takes segment-optimal inputs $C_1$ and $C_2$
and returns a segment-optimal circuit that is
functionally equivalent to the concatenation of the input circuits.

Given that the inputs $C_1$ and $C_2$ are segment optimal,
all $\Omega$-segments that lie completely within $C_1$ or $C_2$
are already optimal.
Therefore, the function only considers and optimizes
``boundary segments''
which have some layers from circuit $C_1$ and other layers from circuit $C_2$.

To optimize segments at the boundary,
the function creates a ``super segment'', named $W$,
by concatenating the last $\Omega$ layers of circuit $C_1$
with the first $\Omega$ layers of circuit $C_2$.
The function denotes this concatenation as
$C_1[d_1 - \Omega : d_1] + C_2[0 : \Omega]$ (see \lineref{combine}).
% where $C_1[d_1 - \Omega : d_1]$ represents the last $\Omega$ layers of circuit $C_1$
% and $C_2[0 : \Omega]$ represents the first $\Omega$ layers of circuit $C_2$.
The meld function calls the oracle on $W$
and retrieves the $W'$,
which is guaranteed to be segment-optimal because
it is returned by the oracle.
The meld function then considers the costs of $W$ and $W'$.

If the costs of $W$ and $W'$ are identical, then $W$ is already
segment optimal.  Consequently, all $\Omega$-segments at the boundary
of $C_1$ and $C_2$ are also optimal.
The key point is that the ``super segment'' $W$ encompasses
all possible $\Omega$-segments at the boundary of $C_1$ and $C_2$.
To see this,
let's choose an $\Omega$-segment at boundary,
which takes the last $i > 0$ layers of circuit $C_1$
and the first $j > 0$ layers from of circuit $C_2$;
we can write this as $C_1[d_1 - i : d_1] + C_2[0 : j]$,
where $d_1$ is the number of layers in $C_1$.
Given that this is an $\Omega-$segment and has $i + j$ layers,
we get that $i + j = \Omega$ and $i < \Omega$ and $j < \Omega$.
Now observe that our chosen segment $C_1[d_1 - i : d_1] + C_2[0 : j]$
is contained within the super segment $W = C_1[d_1 - \Omega : d_1] + C_2[0 : \Omega]$ (\lineref{combine}),
because $i < \Omega$ and $j < \Omega$.
Given that $W$ is segment optimal, our chosen segment is also optimal (relative to the oracle).

Returning to the \textsf{meld} algorithm,
consider the case where
the segment $W'$ improves upon the segment $W$.
In this case, meld incorporates $W'$ into the circuit
and propagates this change to the neighboring layers.
To do this, meld works with three segment optimal circuits:
circuit $C_1[0 : d_1 - \Omega]$,
which contains the first $d_1 - \Omega$ layers of circuit $C_1$,
is segment optimal because $C_1$ is segment optimal;
the circuit $W'$ is segment optimal because it was returned by the oracle;
and the circuit $C_2[\Omega: d_2]$,
which contains the last $d_2 - \Omega$ layers of circuit $C_2$,
is segment optimal because $C_2$ is segment optimal.
Thus, we propagate the changes of window $W'$,
by recursively melding these segment optimal circuits.

In \figref{lopt-code},
the function meld first melds the remaining layers of circuit $C_1$ with the segment $W'$,
obtaining circuit $M$ (see \lineref{mrec}),
and then melds the circuit $M$ with the remaining layers of $C_2$.

\subsection{Meld Example}
We present an example of how \lstinline{meld} joins two circuits by
optimizing from the ``seam'' out, and does so ``lazily'',  as needed.

\begin{figure}
  % \begin{minipage}[b]{0.5\linewidth}
  %     \includegraphics[width=1.3\linewidth, left]{media/transforms.pdf}
  % \end{minipage}%
  % \begin{minipage}[b]{0.5\linewidth}
  %   \includegraphics[width=1.3\linewidth]{media/circuits.pdf}
  % \end{minipage}
  \includegraphics[width=\linewidth]{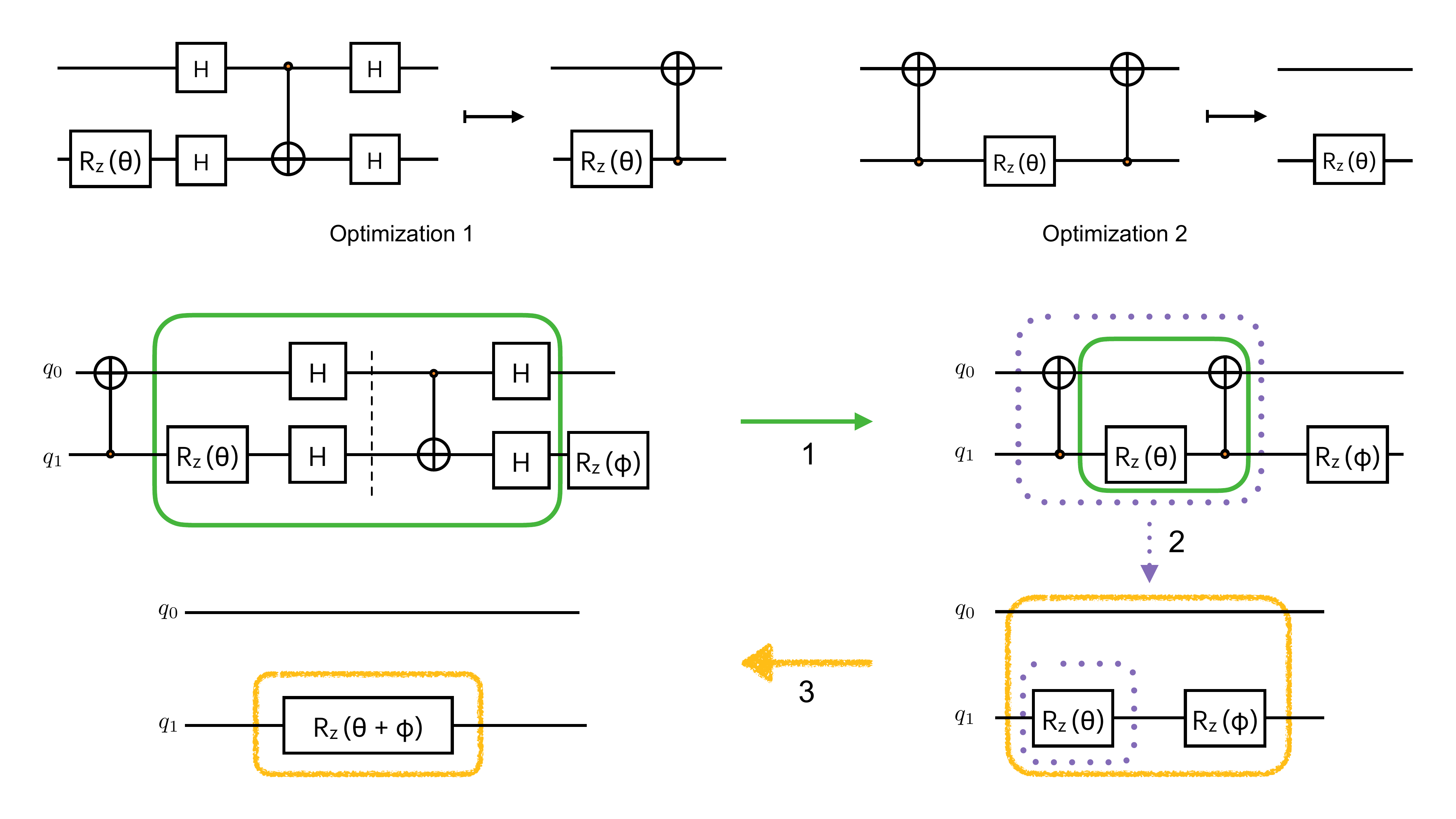}
  \caption{
  The figure shows a three-step meld operation and illustrates
  how it propagates optimizations at the boundary of two optimal circuits.
  At the top, the figure shows specific optimizations ``1'' and ``2'',
  and towards the bottom,
  the figure shows the optimization steps of meld.
  Before the first step,
  two individual circuit segments that are optimal are separated by a dashed line.
  Each meld step considers a segment, represented by a box with solid/dotted/shaded lines,
  and applies an optimization to it, reducing the gate count.
  The first step focuses on the boundary segment within the solid green box,
  overlapping with both circuits,
  and applies ``Optimization 1''.
  This step introduces a flipped $\mathsf{CNOT}$ gate,
  which interacts with a neighboring $\mathsf{CNOT}$ gate
  and triggers ``Optimization 2'' in the purple dotted box.
  The third step merges two neighboring rotation gates in the yellow shaded box.
  }
  \label{fig:propagate}
\end{figure}

\figref{propagate} shows a three-step meld operation
that identifies optimizations at the boundary of two circuits.
All the circuits in the figure are expressed using
the $\mathsf{H}$ gate (Hadamard gate),
the $\mathsf{R_Z}$ gate (rotation around $\mathsf{Z}$),
and the two-qubit $\mathsf{CNOT}$ gate (Controlled Not gate),
which is represented using a dot and an XOR symbol.
We first provide background on the optimizations used by meld
and label them ``Optimization 1'' and ``Optimization 2''.
Optimization 1 shows that when a $\mathsf{CNOT}$ gate is surrounded by
four $\mathsf{H}$ gates,
all of these gates can replaced by a single
$\mathsf{CNOT}$ gate whose qubits are flipped.
Optimization 2 shows that when two $\mathsf{CNOT}$ gates
are separated by a $\mathsf{R_Z}$ gate as shown, they may be removed.

The steps in the figure describe a meld operation
on the two circuits separated by a dashed line, which represents their seam.
To join the circuits, the meld operation proceeds outwards in both directions
and
and optimizes the boundary segment, represented as a green box with solid lines.
The meld applies Optimization 1 to the green segment
and this introduces a flipped $\mathsf{CNOT}$ gate.

The meld propagates this change in step 2,
by considering a new segment,
which includes a neighboring layer.
We represent this segment by a purple box with dotted lines,
and it contains two $\mathsf{CNOT}$ gates,
one of which was introduced by the first optimization.
The meld then applies Optimization 2,
removing the $\mathsf{CNOT}$ gates and
bringing the two rotation gates next to each other.
Note that Optimization 2 became possible only because of Optimization 1,
which introduced the flipped the $\mathsf{CNOT}$ gate.
In the final step,
the meld considers the segment represented by a yellow box with shaded lines
and performs a third optimization, merging the two rotation gates.
Overall,
this sequence of optimizations,
at the boundary of two circuits,
reduces the gate count by seven.

\subsection{Correctness and Efficiency}

Because our algorithm cuts the input circuit into subcircuits and
optimizes them independently, it is far from obvious that its output is
segment optimal.
We prove that this is indeed the case with \thmref{opt} below.
The reason for this is the meld operation that is able to optimize
circuit cuts.
We also prove, with \corref{linear-calls}, that even though
meld behaves dynamically and its cost varies from one circuit to
another, it remains efficient, in the sense that the number of calls
to oracle is always linear in the size of the circuit plus the
improvement in the cost.

The proofs for these are presented in the provided appendix.

\begin{lem}[Segment optimality of meld]
	\label{lem:meld-is-optimal}
	Given any additive \cost{} function and any segment optimal circuits
	$C_1$ and $C_2$, the result of \lstinline{meld}$(C_1, C_2)$ is a
	segment optimal circuit $C$ and $\costof{C} \leq \costof{C_1} + \costof{C_2}$.
\end{lem}

\begin{theorem}[Segment optimality algorithm] \label{thm:opt}
  For any circuit $C$, the function $\mathsf{segopt}(C)$ outputs
  a segment optimal circuit.
  \end{theorem}

\begin{theorem}[Efficiency of segment optimization] \label{thm:cost}
  % Consider a segment length $\Omega$, an oracle that optimizes for the cost function $\mathbf{cost}$,
  % and a circuit $C$.
  The function $\mathsf{segopt}(C)$ calls the oracle at most
  $\mathsf{length}(C) + 2\Delta$ times on segments of length at most $2\Omega$,
  where $\Delta$ is the improvement in the cost of the output.
  \end{theorem}

\begin{corollary}[Linear calls to the oracle]
  When optimizing for gate count, our $\mathsf{segopt}(C)$ makes a linear,
  $O(\mathsf{length}(C) + \sizeof{C})$, number of calls to the oracle.
\label{cor:linear-calls}
\end{corollary}

We experimentally validate this corollary in \secref{eval-calls},
where we study the number of oracle calls made by our algorithm
for many circuits.

\begin{wrapfigure}{r}{0.33\textwidth}
\definecolor{mPurple}{rgb}{0.58,0,0.82}
\lstset{
  basicstyle=\footnotesize\fontfamily{ttfamily}\selectfont, % set the font
  keywordstyle={\color{mPurple}}, % set the keyword style
  morecomment=[l]{//},
  commentstyle=\rmfamily\slshape,
  % commentstyle=\itshape, % set the comment style
  showstringspaces=false, % don't show spaces in strings
  columns=fullflexible, % use proportional spacing
  morekeywords={fun,func,let,val,in,end,case,of,SOME, NONE, and, structure, if, or, else, then, return, def}, % define additional keywords
  mathescape=true, % enable math mode
  escapechar={@},
  keepspaces=true,
  breaklines=true,
  numbers=none,
  numbersep=0pt,
  xleftmargin=0em}
\begin{lstlisting}
def $\mathsf{\algname{}}^*(f, C)$:
  $C'$ = segopt(compact$(C)$)
  if $1 - \frac{\costof{C'}}{\costof{C}} \leq f$:
    return $C'$
  else:
    return $\mathsf{\algname{}}^*(f, C')$
\end{lstlisting}
\caption{$\algname{}^*$ terminates when
the cost improvement ratio falls below the \emph{convergence threshold}, $f$.}
\label{fig:oac-star}
\end{wrapfigure}
\subsection{\algname{}*: controlling convergence}
\label{sec:convergence-thresh}
We observe that, except for the very last round,
each round of our \algname{} algorithm improves the cost of the circuit.
This raises a practical question: how does the
improvement in cost vary across rounds?
For the vast majority of our evaluation,
we observed that nearly all optimization (> 99\%) occur in the
first round itself (see \secref{converge});
the subsequent rounds have a small impact on the quality.
% which suggests that local optimality and segment optimality are quite similar
% in practice.
%
Based on this observation,
we propose $\mathsf{\algname{}}^*$ which uses a
convergence threshold $0 \leq \epsilon \leq 1$
to provide control over how quickly the algorithm converges.

The $\algname{}^*$ algorithm, in \figref{oac-star},
terminates as soon as the cost is reduced
by a smaller fraction than $f$.
For example, if we measure cost as the number of gates and $\epsilon = 0.01$,
then $\algname{}^*$ will terminate as soon as an
optimization round removes fewer than $1\%$ of the remaining gates.
Note that $\algname{}^*$ gives two guarantees:
1) the output circuit is segment optimal,
and 2) the fractional cost improvement in the last round is less than $\epsilon$.
In addition, setting $\epsilon = 0$
results in identical behavior to \algname{} and guarantees local optimality.

\section{Evaluation} \label{sec:eval}

We perform an empirical evaluation of the effectiveness of the
local optimality approach for quantum circuits.
Specifically, we consider the following research questions (RQ).
\begin{description}
    % \item[Question I:] Can the \algname algorithm be implemented with reasonable constant factors?
    % \item[Question II:] Can the implementation use existing different optimizers as oracles?
  \item[RQ I:] Is local optimality and the \algname{} algorithm
    effective in terms of efficiency, scalability, and optimization
    quality?

 \item[RQ II:] Is empirical performance consistent with the asymptotic bound?

 \item[RQ III:] What is the role of the lazy \textsf{meld} operation? 
 
 \item[RQ IV:] What is the impact of segment size $\Omega$ and
 compaction on the local optimality and performance of \algname{} algorithm?
\end{description}

To answer these questions,
we implement the \algname{} algorithm and integrate
it with \voqc{}~\cite{hietala2021verified} as an oracle.
We chose \voqc{} because it is overall the best optimizer both in
terms of efficiency and quality of optimization among all the
optimizers that we have experimented with.
%
%two different optimizers, \voqc{}~\cite{hietala2021verified} and \feyntool{}~\cite{amy2019formal}, as oracles.
%
%We evaluate the effectiveness of local optimality and \algname{}, on two gate sets
%and compare it with four state-of-the-art optimizers:
%\quartz{}, \queso{}, \voqc{}, and \feyntool{}~\cite{quartz-2022, queso-2023,hietala2021verified,amy2019formal}.
%
We evaluate the effectiveness of local optimality and \algname{}, on
the Nam gate set~\cite{Nam_2018} and compare it with three
state-of-the-art optimizers \quartz{}, \queso{}, and
\voqc{}~\cite{quartz-2022, queso-2023,hietala2021verified}.

In brief, these experiments show that our cut-and-meld algorithm
delivers fast optimization while closely matching (within 0.1\%) or
improving the optimization quality for all circuits.
These results show that the local optimality approach can be effective in
optimizing large quantum circuits, and can help scale existing optimizers. 

\iffull
In \appref{clifft}, 
\else
In the Appendix, 
\fi
we present results for
the Clifford+T gate set by using the \feyntool{}~\cite{amy2019formal},
as an oracle.
We omitted these from the main body of the paper due to space reasons
but note that they are similar to the results presented here in terms
of efficiency and quality.

\subsection{Implementation}
\label{sec:impl}

To evaluate whether the \algname{} algorithm  (\secref{algorithm})
is practically feasible,
we implemented \algname{}  in SML
(Standard ML), which comes with an optimizing compiler, MLton, that
can generate fast executables.
Our implementation closely follows the algorithm description.
It uses the layered circuit representation and
represents circuits as an array of arrays, where
each array denotes a ``layer'' of the circuit.
The implementation splits and joins circuit segments by
splitting and joining the corresponding arrays,
performing rounds of optimization and compaction.

As described in \secref{convergence-thresh}, our implementation
allows user control over convergence through
a specified convergence ratio $\epsilon$, where $0 \le \epsilon \le 1$.
In the evaluation, we choose $\epsilon = 0.01$, and analyze this choice in
\secref{compaction}.
% and 
%\iffull
%\appref{converge}.
%\else
%the Appendix.
%\fi
Note that regardless of $\epsilon$,
the implementation always guarantees that output is \wopttext.

Our implementation is parametric in the oracle being used. To allow calls to existing optimizers,
we use MLton's foreign function interface, which supports cross-language calls to C++.
Specifically, to use an existing optimizer as an oracle, we only need to provide a C++ wrapper that takes a circuit in QASM format as input
and returns an optimized circuit as output.

\subsection{Benchmarks and gate set}
% \myparagraph{Benchmarks.}
\label{sec:benchmarks}
To evaluate our \coam{} algorithm,
we consider a benchmark suite of eight circuit families that include both near-term and future fault-tolerant quantum algorithms.
For each family, we select circuits with different sizes by changing the number of qubits.
Our benchmark suite includes advanced quantum algorithm such as
Grover's algorithm~ for unstructured search~\cite{grover1996fast},
the HHL algorithm for solving linear systems of equations~\cite{harrow2009quantum},
Shor's algorithm for factoring large integers~\cite{shor1994algorithms},
and
the Binary Welded Tree (bwt) quantum walk algorithm~\cite{childs2003exponential}.
In addition, our benchmarks include  near-term algorithms like Variational Quantum Eigensolver (vqe)~\cite{peruzzo2014variational}
and reversible arithmetic algorithms~\cite{amy2014polynomial,Nam_2018}
such as boolean satisfaction problems (boolsat) and square-root algorithm (sqrt).

%These circuit are generated using a combination of
%PennyLane~\cite{bergholm2018pennylane}, Qiskit~\cite{qiskit-2019} and NWQBench~\cite{li2021qasmbench}.
%
%
%They also require a large numbers of gates,
%which allows us to evaluate the scalability of different optimizers.

% For each circuit family, we generate circuits of different sizes by
% appropriately changing the number of qubits.
%
% To generate the suite, we use PennyLane~\cite{bergholm2018pennylane} or Qiskit~\cite{qiskit-2019}, depending on the benchmark.
%We evaluate our optimizer on two different gate sets.
%
The benchmark suite is written in the Nam gate set~\cite{Nam_2018}, which consists of the
%The Nam gate set consists of the
% The Nam gate set contains the
Hadamard ($\mathsf{H}$),
Pauli-X ($\mathsf{X}$),
controlled-NOT ($\mathsf{CNOT}$),
and Z-rotation ($\mathsf{R_Z}$) gates~\cite{Nam_2018}.
%
%
%For the Nam gate set,
We preprocess all our benchmarks with the \quartz{} preprocessor,
%in the Nam gate set,
which merges rotation gates~\cite{quartz-2022}.
%

%
% results with \quartz{} and \queso{} could apply
% to other gate sets, including the IBM gate set,
% the Rigetti gate set, and the Ion gate set.
% %
% This is because \quartz{} and \queso{} have shown competitive performance
% for all these gate sets~\cite{quartz-2022, queso-2023},
% and their performance
% would translate to our setting because we use them for all our optimizations.

% we use the benchmark suite provided by the authors of \feyntool{}~\cite{amy2019formal}.
% These benchmarks are primarily expressed in the Clifford+T gate set,
% but also contain other gates such as the Toffoli gate (\lstinline{CCX} gate).
% %
% We also port some of the Nam benchmarks to the Clifford+T gate set using Qiskit.

% Specifically,
% when \coamwith{\quartzt{t}} calls the oracle \quartz{}
% for a circuit containing $k$ gates,
% it gives \quartz{} a timeout of $k*t$.
% %
% The timeout resets after each optimization, meaning
% that \quartzt{t} can run for any amount of time as long as it finds optimizations,
% but it stops when it has not discovered an optimization in $k*t$ time.
% %
% With this policy,
% the \quartzt{t} returns an optimal circuit,
% in the sense that calling \quartzt{t} will not find more optimizations.
% %

%

% \subsection{Question II: Using Different Oracles}
\subsection{RQ I: Effectiveness of \algname{} and local optimality}
\label{sec:scalability}

%We validate the effectiveness of \algname{} on the Nam gate set.
%We validate the effectiveness of \algname{} on two different gate
%sets: Nam and Clifford+T.
%
%In this section, we evaluate with the Nam gate set.
%and \secref{clifft} considers the Clifford+T gate set.
%
To evaluate the effectiveness of \algname{}, we use our \algname{}
implementation with \voqc{} as the oracle on segments of size $\Omega
= 40$ and compare it to optimizers \quartz{}, \queso{}, and
\voqc{}.
The approach works for many different settings of $\Omega$ and we
analyze the impact of $\Omega$ in \secref{var-segment} in detail.
We give each optimizer a 12-hour cut-off time (excluding time for parsing and printing),
to allow completion of the experiments within a reasonable amount of time.
Throughout, we omit circuit-parsing time for timings of \voqc{}, whose parser
%,
% which maps input circuits into the internal representation used for optimization,
appears to scale superlinearly and can take significant time
(sometimes more than the optimization itself).
This approach is consistent with prior work on \voqc{}, which also excludes parse time.
When we use \voqc{} as an oracle of \algname{}, however,
we do include the parse time.
This makes the comparison somewhat unfair for \algname{}.
%

%% Note however that \algname{} only uses the \voqc{} oracle to
%% (parse and) optimize small segments of the circuit, which limits
%% the parsing overhead.
%
% Although this makes the comparison unfair against \algname{},
% if we optimized this overhead away, \algname{} would perform
% the same or better.
% This makes the comparison somewhat unfair to our \algname{},
% but nevertheless we observe good performance in our experiments.
%
% \algname{} calls \voqc{} to (parse and)
% optimize only small segments of the circuit.
%
We evaluate the running times of these optimizers
on benchmarks from the Nam gate set with sizes ranging
from thousands to hundreds of thousands of gates.

\begin{figure}[!ht]
    \centering
%    \small
\begin{tabular}{cccccccc}
                          &        &            & \multicolumn{4}{c}{Time (s)}                      &                                                           \\\cmidrule(lr){4-7}
Benchmark                 & Qubits & Input Size & Quartz & Queso         & VOQC    & \algname{}         & \begin{tabular}[c]{@{}c@{}}\algname{}\\ speedup\end{tabular} \\\midrule{}
\multirow{4}{*}{boolsat}  & 28     & 75670      & 12h    & 12h           & 68.6    & \textbf{45.2}   & 1.52                                                      \\
                          & 30     & 138293     & 12h    & 12h           & 307.3   & \textbf{98.7}   & 3.11                                                      \\
                          & 32     & 262548     & 12h    & 12h           & 1266.2  & \textbf{213.6}  & 5.93                                                      \\
                          & 34     & 509907     & 12h    & 12h           & 6151.0  & \textbf{462.8}  & 13.29                                                     \\\midrule{}
\multirow{4}{*}{bwt}      & 17     & 262514     & 12h    & 12h           & 8303.1  & \textbf{524.9}  & 15.82                                                     \\
                          & 21     & 402022     & 12h    & 12h           & 23236.8 & \textbf{1062.1} & 21.88                                                     \\
                          & 25     & 687356     & 12h    & 12h           & >12h & \textbf{2341.7} & > 18.45                                                     \\
                          & 29     & 941438     & 12h    & 12h           & >12h & \textbf{3982.6} & > 10.85                                                     \\\midrule{}
\multirow{4}{*}{grover}   & 9      & 8968       & 12h    & 12h           & 9.3     & \textbf{4.8}    & 1.94                                                      \\
                          & 11     & 27136      & 12h    & 12h           & 106.5   & \textbf{20.8}   & 5.12                                                      \\
                          & 13     & 72646      & 12h    & 12h           & 815.7   & \textbf{68.1}   & 11.97                                                     \\
                          & 15     & 180497     & 12h    & 12h           & 5743.2  & \textbf{223.9}  & 25.65                                                     \\\midrule{}
\multirow{4}{*}{hhl}      & 7      & 5319       & 12h    & 12h           & \textbf{0.3}     & 0.9    & 0.27                                                      \\
                          & 9      & 63392      & 12h    & 12h           & 74.1    & \textbf{22.9}   & 3.24                                                      \\
                          & 11     & 629247     & 12h    & 12h           & 14868.8 & \textbf{434.3}  & 34.24                                                     \\
                          & 13     & 5522186    & 12h    & Parsing Error & >12h & \textbf{8243.1}       & > 5.24                                                          \\\midrule{}
\multirow{4}{*}{shor}     & 10     & 8476       & 12h    & OOM           & 8.8     & \textbf{5.2}    & 1.70                                                      \\
                          & 12     & 34084      & 12h    & 12h           & 179.9   & \textbf{26.3}   & 6.84                                                      \\
                          & 14     & 136320     & 12h    & 12h           & 3638.4  & \textbf{126.0}  & 28.88                                                     \\
                          & 16     & 545008     & 12h    & 12h           & 70475.2 & \textbf{648.9}  & 108.60                                                    \\\midrule{}
\multirow{4}{*}{sqrt}     & 42     & 79574      & 12h    & 12h           & \textbf{30.0}    & 81.4   & 0.37                                                      \\
                          & 48     & 186101     & 12h    & 12h           & \textbf{191.2}   & 268.0  & 0.71                                                      \\
                          & 54     & 424994     & 12h    & 12h           & 3946.5  & \textbf{679.8}  & 5.81                                                      \\
                          & 60     & 895253     & 12h    & 12h           & >12h & \textbf{1653.5} & > 26.13                                                     \\\midrule{}
\multirow{4}{*}{statevec} & 5      & 31000      & 12h    & OOM           & \textbf{1.6}     & 4.3    & 0.38                                                      \\
                          & 6      & 129827     & 12h    & 12h           & 45.9    & \textbf{27.1}   & 1.70                                                      \\
                          & 7      & 526541     & 12h    & 12h           & 1812.2  & \textbf{164.7}  & 11.00                                                     \\
                          & 8      & 2175747    & 12h    & 12h           & >12h & \textbf{1345.1} & > 32.12                                                     \\\midrule{}
\multirow{4}{*}{vqe}      & 12     & 11022      & 12h    & 12h           & \textbf{0.2}     & 1.2    & 0.13                                                      \\
                          & 16     & 22374      & 12h    & 12h           & \textbf{0.6}     & 3.4    & 0.18                                                      \\
                          & 20     & 38462      & 12h    & 12h           & \textbf{2.0}     & 7.0    & 0.29                                                      \\
                          & 24     & 59798      & 12h    & 12h           & \textbf{5.4}     & 13.4            & 0.41                                                      \\\midrule{}
avg                       &        &            &        &               &         &                 & > 12.62
\end{tabular}

\caption{The figure shows the running time in seconds of the four optimizers, using gate count as the cost metric.
The column "\algname{} Speedup" is the speed of our  \algname{} with respect to \voqc{}, calculated as \voqc{} time divided by \algname{} time.
These measurements show that our optimizer \algname{} can be significantly faster, especially for large circuits (more than one order of magnitude on average).  As \figref{main-gate-count} shows, these time improvements come without any loss in optimization quality.
These results suggest that local optimality approach to optimization of quantum circuits can be effective in practice.}
    \label{fig:main-time}
\end{figure}
\begin{figure}[!ht]
  \centering
  \small
\medskip
\begin{tabular}{ccccccc}
                          &        &            & \multicolumn{4}{c}{Optimizer}         \\\cmidrule(lr){4-7}
Benchmark                 & Qubits & Input Size & Quartz  & Queso   & VOQC    & \algname{} \\\midrule{}
\multirow{4}{*}{boolsat}  & 28     & 75670      & -41.1\% & -30.4\% & -83.2\% & \textbf{-83.7\%} \\
                          & 30     & 138293     & -23.5\% & -30.2\% & -83.3\% & \textbf{-83.7\%} \\
                          & 32     & 262548     & -6.9\%  & 0.0\%   & -83.3\% & \textbf{-83.5\%} \\
                          & 34     & 509907     & -2.9\%  & 0.0\%   & -83.3\% & \textbf{-83.4\%} \\\midrule{}
\multirow{4}{*}{bwt}      & 17     & 262514     & -8.7\%  & -0.1\%  & -30.0\% & \textbf{-31.1\%} \\
                          & 21     & 402022     & -3.1\%  & -0.1\%  & -38.4\% & \textbf{-40.0\%} \\
                          & 25     & 687356     & -0.8\%  & 0.0\%   & N.A.    & \textbf{-43.8\%} \\
                          & 29     & 941438     & -0.4\%  & 0.0\%   & N.A.    & \textbf{-44.5\%} \\\midrule{}
\multirow{4}{*}{grover}   & 9      & 8968       & -9.4\%  & -13.7\% & \textbf{-29.4\%} & \textbf{-29.4\%} \\
                          & 11     & 27136      & -9.5\%  & -9.6\%  & -29.9\% & \textbf{-30.0\%} \\
                          & 13     & 72646      & -9.6\%  & -0.3\%  & \textbf{-29.7\%} & \textbf{-29.7\%} \\
                          & 15     & 180497     & -9.6\%  & -0.2\%  & \textbf{-29.5\%} & \textbf{-29.5\%} \\\midrule{}
\multirow{4}{*}{hhl}      & 7      & 5319       & -26.6\% & -28.7\% & \textbf{-55.4\%} & -55.3\% \\
                          & 9      & 63392      & -24.4\% & -23.8\% & -56.3\% & \textbf{-56.5\%} \\
                          & 11     & 629247     & -1.1\%  & 0.0\%   & \textbf{-53.7\%} &\textbf{ -53.7\%} \\
                          & 13     & 5522186    & 0.0\%   & N.A.    & N.A.    & \textbf{-52.6\%} \\\midrule{}
\multirow{4}{*}{shor}     & 10     & 8476       & 0.0\%   & -5.4\%  & \textbf{-11.1\%} & -11.0\% \\
                          & 12     & 34084      & 0.0\%   & -3.8\%  & \textbf{-11.2\%} & \textbf{-11.2\%} \\
                          & 14     & 136320     & 0.0\%   & 0.0\%   & \textbf{-11.3\%} & -11.2\% \\
                          & 16     & 545008     & 0.0\%   & 0.0\%   & \textbf{-11.3\%} &\textbf{ -11.3\%} \\\midrule{}
\multirow{4}{*}{sqrt}     & 42     & 79574      & -16.2\% & -0.1\%  & \textbf{-33.0\%} &\textbf{ -33.0\%} \\
                          & 48     & 186101     & -15.2\% & 0.0\%   & \textbf{-32.7\%} & -32.6\% \\
                          & 54     & 424994     & -5.1\%  & 0.0\%   & \textbf{-32.4\%} & -32.3\% \\
                          & 60     & 895253     & -2.1\%  & 0.0\%   & N.A.    & \textbf{-34.3\%} \\\midrule{}
\multirow{4}{*}{statevec} & 5      & 31000      & -48.4\% & -74.5\% & -78.8\% & \textbf{-78.9\%} \\
                          & 6      & 129827     & -35.0\% & -29.7\% & \textbf{-78.4\%} & \textbf{-78.4\%} \\
                          & 7      & 526541     & -2.0\%  & -29.7\% & \textbf{-78.1\%} & \textbf{-78.1\%} \\
                          & 8      & 2175747    & -0.1\%  & 0.0\%   & N.A.    & \textbf{-78.7\%} \\\midrule{}
\multirow{4}{*}{vqe}      & 12     & 11022      & -35.4\% & -69.1\% & -63.0\% & \textbf{-69.5\%} \\
                          & 16     & 22374      & -33.7\% & -64.2\% & -60.1\% &\textbf{ -66.3\%} \\
                          & 20     & 38462      & -32.2\% & -60.8\% & -57.4\% & \textbf{-63.4\%} \\
                          & 24     & 59798      & -30.8\% & -36.4\% & -54.9\% & \textbf{-60.6\%} \\\midrule{}
avg                       &        &            & -13.6\% & -16.5\% & -48.1\% & \textbf{-49.4\%}
\end{tabular}
\caption{Optimization Quality.  The figure shows percentage reduction in gate count achieved by the four optimizers. 
We write ``N.A.'' in cases where the optimizer did not finish within the allotted 12 hour deadline.
The experiments show that our cut-and-meld optimizer \algname{}
optimizes slightly better on average than all other optimizers.
%
%% Given that our \algname{} optimizer is faster than all other optimizers (\figref{main-time}), these experiments show that local optimality approach can improve performance without any loss in optimization quality.  
%% %
%% These results suggest that local optimality approach to optimization of quantum circuits can be effective in practice.
}
    \label{fig:main-gate-count}
\end{figure}

\myparagraph{Time Performance.}
\figref{main-time} show the time for our \algname{} implementation (with \voqc{} oracle) compared against \quartz{}, \queso{}, and \voqc{}.
The figure includes eight families of circuits, where
horizontal lines separate families and circuits within each family
arranged body increasing qubit and gate counts.
The optimizers \quartz{} and \queso{} use the maximum allotted time of 12 hours in all circuits,
because they explore a very large search space of all optimizations.
In a few cases, \queso{} throws an error or runs out of memory (denoted ``OOM'').
The \voqc{} optimizer and our \algname{} optimizer terminate much faster.
Specifically, \algname{} optimizes all circuits
between 0.2 seconds and 3 hours depending on the size,
and \voqc{} finishes for all but six benchmarks within 12 hours.
In the figure, we highlight in bold the fastest optimizer(s) for each circuit.

Comparing between \voqc{} and our \algname{}, we observe the following:
\begin{itemize}
    \item \textbf{Performance:} \algname{} is the fastest across the board except for \texttt{vqe}
    and except perhaps for the smallest circuits in some families.
    \item \textbf{Scalability:} the gap between \algname{} and \voqc{} increases as the circuit size increases,
    with \algname{} performing as much as 100$\times$ faster in some cases.
    \item \textbf{Overall:} \algname{} is over an order of magnitude faster than \voqc{} on average.
\end{itemize}
In the case of the \texttt{vqe} family, \voqc{} is consistently faster, but as we discuss next,
this comes at the cost of poorer optimization quality.
For the small circuits of families \texttt{hhl}, \texttt{statevec}, and \texttt{sqrt},
our optimizer is slower than \voqc{}.
This is due to the overheads that our implementation incurs for
(1) splitting and joining circuits,
(2) serialization/deserialization of input/output circuits for each oracle call, and
(3) various system-level calls needed to support calls to an external oracle.
For example,
for the $7$-qubit \texttt{hhl} benchmark and the $42$-qubit \texttt{sqrt} benchmark, we have measured that at least 30\% of the running
time is spent parsing and serializing/deserializing circuits.

\myparagraph{Optimization quality.}
Our experiments so far show that our \algname{} performs well but it
does not give evidence of optimization quality.
%
% This raises the question: although \algname{} finds local optimizations
% in an efficient and scalable fashion, does it lose anything in terms
% of quality by only considering local optimizations
% and not operating on the entire circuit simultaneously?
%
%% TODO: the quality improvemenst are not significant
%% so, don't mention
%and in fact, \algname{} produces better quality circuits in some cases.
%
%This is due to the fact \algname{} guarantees local optimality when it terminates.
%
\figref{main-gate-count} shows
the output quality (measured by gate count)
of all optimizers for eight families of circuits.
%
% Each family is separated by horizontal lines,
% with circuits arranged by increasing qubit and gate counts.
% %
The figure shows the original gate count and the percent reduction
in gate count achieved by tools \quartz{}, \queso{}, \voqc{}, and \algname{}.
The best optimizers are highlighted in bold.
These results show that \algname{} always matches the best optimizer
within $0.1\%$ or outperforms it.
On average, \algname{} reduces the gate count by $49.7\%$, improving by
$1\%$ over the second best.
We note all optimizers except for our \algname{}, are unable to finish
some large circuits within the allotted 12-hour time limit or yield
very small (less than 1\%) improvement.
We present a more detailed discussion of these experiments below.

The results show that \algname{} and \voqc{} produce overall
better circuits than \quartz{} and \queso{}.
For the \texttt{hhl} family, both \algname{} and \voqc{} achieve
reductions of around $56\%$, while \quartz{} and \queso{} are around
$26\%$ for $7$ and $9$ qubits, and less than $1\%$ for $11$ qubits.
In the \texttt{statevec} family,
\algname{} and \voqc{} consistently reduce the gate count by $78\%$.
However, for the 8-qubit case, \voqc{} does not finish within our timeout
of 12 hours so we write ``N.A.''.
\queso{} also finds comparable reductions for the 5-qubit benchmark.

% Internally,
% our \algname{} uses \voqc{} as a subroutine to optimize
% small segments of the circuit, rather than optimizing
% the entire circuit at once.
% %
% Specifically, \algname{} considers the circuit in segments
% of size $\Omega = 40$ and applies \voqc{} to them.
% %
% In contrast,
% when used as a standalone optimizer,
% \voqc{} processes the entire circuit and could theoretically find more optimizations
% by considering all possible gates simultaneously.
% %
% However,
% our experiments show that \algname{}'s segment-based approach
% achieves the same quality.
% %
% This suggests that local optimality,
% as achieved by our optimizer,
% is a good goal for circuit optimization,
% because it does not miss any optimizations in practice.
% %
% As we saw in \figref{main-time},
% the approach is significantly faster because it scales better.
% %

For almost all families, we observe that the output quality of
\algname{} matches that of \voqc{} within $0.1\%$ or improves it,
sometimes significantly.
%
%Our \algname{} uses \voqc{} as the oracle on segments of length $\Omega = 40$
%and produces almost locally optimal circuits
%(upto the convergence ratio $\epsilon$, see \secref{impl}).
%
%
Specifically, for the \texttt{vqe} family, \algname{} optimizes better
than \voqc{}.
For example, on the 24-qubit \texttt{vqe} circuit, \algname{} improves
the gate count by $60.6\%$, and \voqc{} improves the gate count by
$54.9\%$.
Indeed, we observed that running \voqc{} twice by running it again on
its own output circuit bridges this gap.
%
%On average, \algname{} reduces the gate count by $49.7\%$ versus
%$48.1\%$ for \voqc{}.

%% \myparagraph{Summary.}
%% \figref{main-time} and \figref{main-gate-count}
%% show that
%% our \algname{} optimizer can be significantly faster, especially for larger circuits, because it scales better, and does so without sacrificing optimization quality.
%% %
%% The experiment thus shows that the local optimality approach can work well, especially for larger circuits.

\if0
is over an order of magnitude faster on average
and
produces circuits that match the best quality.}
%
These results demonstrate that
our \algname{} optimizes
circuits in an efficient and scalable fashion
and that local optimality is an effective
quality criterion.
\fi

\subsection{RQ II: Is empirical performance consistent with the asymptotic bound?}
\label{sec:eval-calls}
\begin{figure}[t]
    \centering
    \includegraphics[width=\columnwidth]{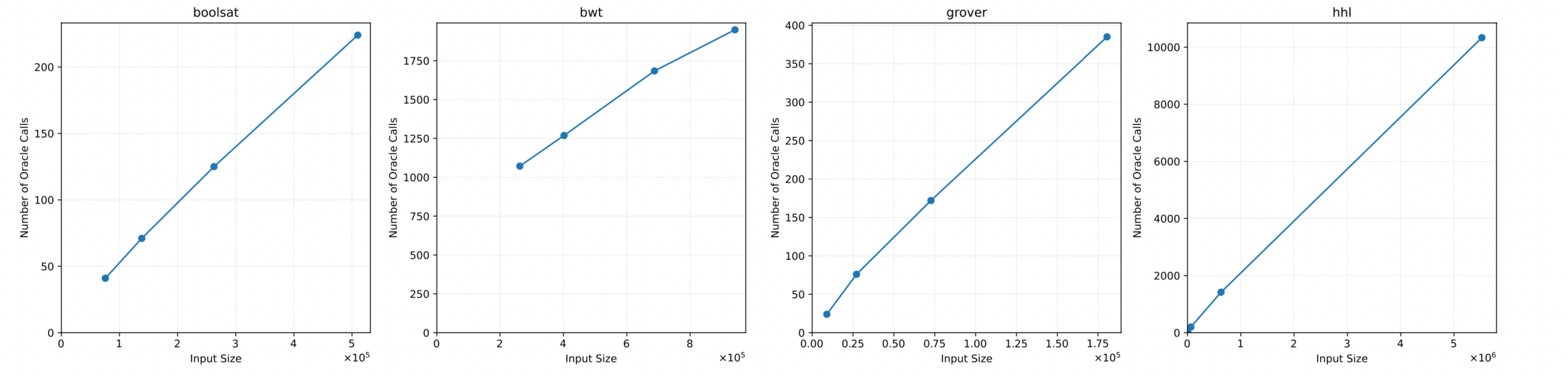}
    \vspace{-20pt}
    \caption{ The number of oracle calls versus input circuit size for
      selected circuits.  The plots show that the number of
      calls scales linearly with the number of gates.  }
    \label{fig:plot-oracle-calls}
\end{figure}

In \secref{algorithm}, we established bounds on the number of oracle
calls performed by our \algname{} algorithm.
In this section, we check that our implementation is consistent with
these bounds by analyzing the number of oracle calls with respect to
the circuit size.
%
%We focus on optimizing for gate count, using \voqc{} as the oracle.
%
\figref{plot-oracle-calls} plots the number of oracle calls made by
our \algname{} optimizer for a subset of circuit families (the other
circuit families behave similarly; see
\iffull
\figref{plot-oracle-calls-all} included in the Appendix).
\else
Appendix).
\fi
The Y-axis represents the number of oracle calls
and the X-axis represents the input circuit size.
The plot shows that the number of oracle calls
increases linearly with circuit size, for all circuit families.

%% This observation is consistent with our asymptotic analysis.
%% %
%% In \corref{linear-calls},
%% we proved that for gate count optimization,
%% each round of our \algname{} performs a linear number of calls to the oracle.
%% %
%% Specifically, we bounded the number of oracle calls by $O(\mathsf{length}(C) + \sizeof{C})$.
%% %
%% Because our optimizer only takes a constant number of rounds to converge,
%% the bounds suggest that the number of oracle calls should be linear
%% in circuit size and depth,
%% which is consistent with the plots in \figref{plot-oracle-calls}.

We note that it would be more desirable to establish that the total
run-time, rather than the number of oracle calls, is linear, but this
is not the case because the oracle optimizers can take asymptotically non-linear time.
For example, the oracle \voqc{} can require at
least quadratic time in the number of gates in the circuit being
optimized, which varies as we increase the qubit
counts.

%% %
%% In these experiments, the oracle that we have used, \voqc{}, is an
%% implementation that primarily aim at verification, rather than
%% efficiency, making empirical analysis of runtime
%% challenging~\cite{hietala-personal-2024}.
%% %
%% We have nevertheless observed that the runtime of the oracle grows
%% roughly quadratically in the number of gates, which is constant for
%% all circuits of a given number of qubits and fixed $\Omega$.

\subsection{RQ III: Ablation Study of Meld}

\begin{figure}[t]
  \scriptsize
  \begin{minipage}[]{0.49\textwidth}
    \begin{tabular}{cccccc|}
      & & & \multicolumn{3}{c}{Optimizer} \\ \cmidrule(lr){4-6}
Family                 & Qubits & Input Size & VOQC  & \algname{} & \hspace{-2pt}\algnameminus{}\hspace{-2pt} \\
 \midrule\multirow{4}{*}{boolsat}  & 28 &   75670 & -83.2\% & \textbf{-83.7\%} & -83.0\% \\
                           & 30 &  138293 & -83.3\% & \textbf{-83.7\%} & -83.1\% \\
                           & 32 &  262548 & -83.3\% & \textbf{-83.5\%} & -83.1\% \\
                           & 34 &  509907 & -83.3\% & \textbf{-83.4\%} & -83.1\% \\
 \midrule\multirow{4}{*}{bwt}      & 17 &  262514 & -30.0\% & \textbf{-31.1\%} & -28.3\% \\
                           & 21 &  402022 & -38.4\% & \textbf{-40.0\%} & -36.3\% \\
                           & 25 &  687356 & N.A.    & \textbf{-43.8\%} & -40.0\% \\
                           & 29 &  941438 & N.A.    & \textbf{-44.5\%} & -41.0\% \\
 \midrule\multirow{4}{*}{grover}   &  9 &    8968 & \textbf{-29.4\%} & \textbf{-29.4\%} & -26.1\% \\
                           & 11 &   27136 & -29.9\% & \textbf{-30.0\%} & -26.3\% \\
                           & 13 &   72646 & \textbf{-29.7\%} & \textbf{-29.7\%} & -26.0\% \\
                           & 15 &  180497 & \textbf{-29.5\%} & \textbf{-29.5\%} & -26.0\% \\
 \midrule\multirow{4}{*}{hhl}      &  7 &    5319 & \textbf{-55.4\%} & -55.3\% & -54.4\% \\
                           &  9 &   63392 & -56.3\% & \textbf{-56.5\%} & -55.6\% \\
                           & 11 &  629247 & \textbf{-53.7\%} & \textbf{-53.7\%} & -53.1\% \\
                           & 13 & 5522186 & N.A.    & \textbf{-52.6\%} & -52.1\% \\
      \midrule
    \end{tabular}
  \end{minipage}
  \begin{minipage}[]{0.49\textwidth}
    \begin{tabular}{cccccc}
      & & & \multicolumn{3}{c}{Optimizer} \\ \cmidrule(lr){4-6}
Family                 & Qubits & Input Size & VOQC  & \algname{} & \hspace{-3pt}\algnameminus{}\hspace{-3pt} \\
 \midrule\multirow{4}{*}{shor}     &  10 &    8476 & \textbf{-11.1\%} & -11.0\% & -10.3\% \\
                           &  12 &   34084 & \textbf{-11.2\%} & \textbf{-11.2\%} & -10.7\% \\
                           &  14 &  136320 & \textbf{-11.3\%} & -11.2\% & -10.8\% \\
                           &  16 &  545008 & \textbf{-11.3\%} & \textbf{-11.3\%} & -10.9\% \\
 \midrule\multirow{4}{*}{sqrt}     & 42 &   79574 & \textbf{-33.0\%} & \textbf{-33.0\%} & -31.5\% \\
                           & 48 &  186101 & \textbf{-32.7\%} & -32.6\% & -31.2\% \\
                           & 54 &  424994 & \textbf{-32.4\%} & -32.3\% & -30.9\% \\
                           & 60 &  895253 & N.A.    & \textbf{-34.3\%} & -32.9\% \\
 \midrule\multirow{4}{*}{statevec} &  5 &   31000 & -78.8\% & \textbf{-78.9\%} & -78.1\% \\
                           &  6 &  129827 & \textbf{-78.4\%} & \textbf{-78.4\%} & -77.9\% \\
                           &  7 &  526541 & \textbf{-78.1\%} & \textbf{-78.1\%} & -77.8\% \\
                           &  8 & 2175747 & N.A.    & \textbf{-78.7\%} & -78.6\% \\
 \midrule\multirow{4}{*}{vqe}      & 12 &   11022 & -63.0\% & \textbf{-69.5\%} & -62.7\% \\
                           & 16 &   22374 & -60.1\% & \textbf{-66.3\%} & -59.7\% \\
                           & 20 &   38462 & -57.4\% & \textbf{-63.4\%} & -57.2\% \\
                           & 24 &   59798 & -54.9\% & \textbf{-60.6\%} & -54.8\% \\
  \midrule
  \end{tabular}
  \end{minipage}
\caption{Ablation study of \lstinline{meld}. The table shows
  percentage reduction in gate count achieved by the version of
  \algname{}, called \algnameminus{}, that uses simply concatenation
  instead of the \lstinline{meld} algorithm, compared with \voqc{} and
  the original \algname{}.
The measurements show
that with the \lstinline{meld} algorithm, our \algname{} optimizer reduces gate counts better than 
\algnameminus{} (-49.4\% vs. -47.3\% on average).
%using \voqc{} alone
%or simply splitting the circuit into segments and optimize them using \voqc{} separately.
}    \label{fig:main-ablation}
\end{figure}

Our algorithm for local optimality optimizes a circuit by cutting it
into smaller subcircuits, optimizing the subcircuits, and melding the
optimized subcircuits into a locally optimal circuit.
If the algorithm joined the circuits together instead of melding them,
then the resulting circuits would not be locally optimal, because the
segments overlapping the circuit cuts may not be optimal.
To understand the impact of the \lstinline{meld} operation, we perform
an ablation experiment.
Specifically, we implement an ablating version of our \algname{},
called \algnameminus{}, that simply concatenates the optimized
subcircuits instead of the \lstinline{meld} operation.

\figref{main-ablation} shows the results of this ablation study.
%
%In this figure, \algnameminus{} splits the circuit into segments, use
%\voqc{} to optimize each segment, and then concatenate the results
%together.
%
%using the same $\Omega=40$ as \algname{}.
%
\algname{} consistently performs better than the ablating version
\algnameminus{}, with over 2\% improvement on the average total gate
count.
Even though the percentage degradation due to ablation may seem modest,
it is significant, because each and every gate has a significant
runtime and fidelity cost on modern and near-term quantum computers.
This ablation study shows the importance of the \lstinline{meld}
algorithm and that of local optimality, which does not hold without
the \lstinline{meld}.
Notably, the quality of the circuits produced by the ablating version
are significantly worse than those produced by the baseline \voqc{}.

\subsection{RQ IV: Impact of compaction and $\Omega$ on the effectiveness of \algname{}}
\label{sec:converge}

\begin{figure}[t]
  \scriptsize
  \begin{minipage}[]{0.4\textwidth}
    \begin{tabular}{ccccc|}
      & & & \multicolumn{2}{c}{Optimizations} \\ \cmidrule(lr){4-5}
      Family & Qubits & \#Rounds & Round 1 & Round 2 \\
      \midrule
      \multirow{4}{*}{boolsat}& 28 & 2 & 100.00\% & 0.00\% \\
      & 30 & 2 & 100.00\% & 0.00\% \\
      & 32 & 2 & 100.00\% & 0.00\% \\
      & 34 & 2 & 100.00\% & 0.00\% \\
      \midrule\multirow{4}{*}{bwt}& 17 & 2 & 99.59\% & 0.41\% \\
      & 21 & 2 & 99.79\% & 0.21\% \\
      & 25 & 2 & 99.83\% & 0.17\% \\
      & 29 & 2 & 99.90\% & 0.10\% \\
      \midrule\multirow{4}{*}{grover}& 9 & 2 & 99.96\% & 0.04\% \\
      & 11 & 2 & 99.90\% & 0.10\% \\
      & 13 & 2 & 99.99\% & 0.01\% \\
      & 15 & 2 & 99.97\% & 0.03\% \\
      \midrule\multirow{4}{*}{hhl}& 7 & 2 & 99.76\% & 0.24\% \\
      & 9 & 2 & 99.95\% & 0.05\% \\
      & 11 & 2 & 99.96\% & 0.04\% \\
      & 13 & 2 & 99.95\% & 0.05\% \\
      \midrule
    \end{tabular}
  \end{minipage}
  \begin{minipage}[]{0.4\textwidth}
    \begin{tabular}{ccccc}
      & & & \multicolumn{2}{c}{Optimizations} \\ \cmidrule(lr){4-5}
      Family & Qubits & \#Rounds & Round 1 & Round 2 \\
    \midrule\multirow{4}{*}{shor}& 10 & 2 & 100.00\% & 0.00\% \\
  & 12 & 2 & 99.97\% & 0.03\% \\
  & 14 & 2 & 99.96\% & 0.04\% \\
  & 16 & 2 & 99.99\% & 0.01\% \\
  \midrule\multirow{4}{*}{sqrt}& 42 & 2 & 99.90\% & 0.10\% \\
  & 48 & 2 & 99.81\% & 0.19\% \\
  & 54 & 2 & 99.97\% & 0.03\% \\
  & 60 & 2 & 100.00\% & 0.00\% \\
  \midrule\multirow{4}{*}{statevec}& 5 & 2 & 100.00\% & 0.00\% \\
  & 6 & 2 & 99.92\% & 0.08\% \\
  & 7 & 2 & 99.92\% & 0.08\% \\
  & 8 & 2 & 99.99\% & 0.01\% \\
  \midrule\multirow{4}{*}{vqe}& 12 & 2 & 99.95\% & 0.05\% \\
  & 16 & 2 & 99.99\% & 0.01\% \\
  & 20 & 2 & 99.98\% & 0.02\% \\
  & 24 & 2 & 99.99\% & 0.01\% \\
  \midrule
  \end{tabular}
  \end{minipage}
  \vspace{-1pt}
    \caption{The number of rounds and the
  percentage of optimizations %found
  in each optimization round of \algname{}.}
  \label{fig:num-rounds}
\end{figure}

\subsubsection{Impact of compaction}
\label{sec:compaction}
\begin{wrapfigure}{r}{0.5\textwidth}
\centering
\small
\vspace{-10pt}
\begin{tabular}{ccc}
$\Omega$ & Time (s) & Output gate count (reduction) \\
\midrule
2 & 17.4 & 27600 (-56.46\%) \\
5 & 13.9 & 27620 (-56.43\%) \\
10 & 17.7 & 27609 (-56.45\%) \\
20 & 18.9 & 27590 (-56.48\%) \\
40 & 21.8 & 27570 (-56.51\%) \\
80 & 35.3 & 27564 (-56.52\%) \\
160 & 66.6 & 27559 (-56.53\%) \\
320 & 122.8 & 27551 (-56.54\%) \\
% grover-n9:
%2 & 3.0 & 6368 (-28.99\%) \\
%5 & 3.3 & 6379 (-28.87\%) \\
%10 & 4.1 & 6343 (-29.27\%) \\
%20 & 3.5 & 6333 (-29.38\%) \\
%40 & 4.5 & 6330 (-29.42\%) \\
%80 & 8.1 & 6330 (-29.42\%) \\
%160 & 14.3 & 6330 (-29.42\%) \\
%320 & 14.6 & 6330 (-29.42\%) \\
\midrule
\voqc{} & 74.1 & 27673 (-56.35\%) \\
\end{tabular}
\caption{Choice of $\Omega$: performance of \algname{} with different $\Omega$
on the \texttt{hhl} circuit with 9 qubits, which initially contains 63392 gates.
For reference, we also present the performance of \voqc{} at the end.}
   \label{fig:omega}
\end{wrapfigure}
% TODO: typeset at a better location
%For all benchmark circuits in \figref{main-gate-count},
%using the convergence ratio $\epsilon = 0.01$,
%\algname{} converges very quickly, always
%terminating after $2$ rounds of optimization.
%Furthermore, it consistently finds over 99\% of the optimizations in the first round.
\figref{num-rounds} shows the number of rounds and percentage
optimizations for each round of our \algname{} algorithm, using the
convergence ratio $\epsilon = 0.01$.
%
%% Similar to \figref{main-gate-count},
%% we use our \algname{} optimizer, with \voqc{} as the oracle
%% and convergence ratio $\epsilon = 0.01$,
%% to optimize the benchmarks in the Nam gate set.
%% %
The results show that \algname{} converges very quickly, always
terminating after $2$ rounds of optimization, and that it
consistently finds over 99\% of the optimizations in the first round.
This is because our \algname{} ensures the slightly weaker segment
optimality after the first round of optimization (see
\secref{algorithm}), which ensures that all segments are optimal,
though there may be gaps.
%We present the detailed result in the Appendix.
%
The experiment shows that although compaction can enable some
optimizations by removing the gaps, its impact on these benchmarks is
minimal.
We separately ran the same experiments with $\epsilon = 0$, which
forces the algorithm to run up to perfect convergence, and observed
that \algname{} requires $4$ rounds of optimization on average over
all circuits.
These results show that in practice a small number of compaction
rounds suffice to obtain results that are within a very small fraction
of the local optimal.

\subsubsection{Impact of varying segment size $\Omega$}
\label{sec:var-segment}

\figref{omega} shows the running time and the gate count reduction of \algname{} with different
values of $\Omega$
on the \texttt{hhl} circuit with 9 qubits, which initially contains 63392 gates.
The results show that for a wide range of $\Omega$ values,
our optimizer produces a circuit of similar quality to the baseline \voqc{}, and typically does so in significantly less time.
When $\Omega$ is large, \algname{}'s running time scales linearly with $\Omega$, and the output gate count reduces marginally when $\Omega$ increases.
We choose $\Omega=40$ in our evaluation to achieve a balance between running time and output quality but note that many different values work similarly well.

\section{Related Work}
We discussed most closely related work in the body of the paper.  In
this section, we present a broader overview of the work on quantum
circuit optimization.

\paragraph{Cost Functions.}
Gate count is a widely used metric for optimizing quantum circuits.
In the NISQ era, reducing gate count improves circuit performance
by minimizing noise from operations and decoherence.
It also reduces resources in fault-tolerant architectures
like the Clifford+T gate set.
Researchers have developed techniques to reduce gate count by
either directly optimizing circuits or
resynthesizing parts using efficient synthesis algorithms.
We cover optimization techniques later in the section.

%\paragraph{Other Gate Sets and Metrics.}

In addition to reducing gate counts,
compilers like Qiskit and t$\ket{\textnormal{ket}}$,
implement circuit transformations that optimize cost specific to NISQ architectures.
Examples include maximizing circuit fidelity in the presence of noise~\cite{murali2019noise, tannu2019not},
and
reducing qubit mapping and routing overhead (SWAP gates)
for specific device topologies~\cite{molavi2022qubit, lye2015determining, itoko2020optimization, li2019tackling},
or hardware-native gates and pulses \cite{nottingham2023decomposing, wu2021tilt, shi2019optimized, gokhale2020optimized}.
Techniques also exist to optimize/synthesis
circuits for specific unitary types, such as classical
reversible gates~\cite{prasad2006data, ding2020square, bandyopadhyay2020post, wille2019towards},
\clifft{}~\cite{amy2020number, kliuchnikov2014asymptotically, ross2014optimal, kissinger2020reducing},
Clifford-cyclotomic~\cite{forest2015exact}, V-basis~\cite{bocharov2013efficient, ross2015optimal},
and Clifford-CS~\cite{glaudell2020optimal} circuits.
While algorithms for small unitaries produce
Clifford+T circuits with an asymptotically optimal number of $\mathsf{T}$ gates~\cite{giles2013remarks},
efficiently generating optimal large Clifford+T circuits remains a challenge.
The \feyntool{} optimizer is used for optimizing the $\mathsf{T}$ count of quantum circuits.
It uses an efficient (polynomial-time) algorithm called phase folding~\cite{amy2014polynomial},
to reduce phase gates, such as the $\mathsf{T}$ gate, by merging them.
More generally, the Feynman toolkit combines phase folding
with synthesis techniques to optimize other metrics like the CNOT count~\cite{amy2019formal}.
We demonstrate that our \algname{} algorithm,
which guarantees local optimality,
can use \feyntool{} as an oracle for optimizing
$\mathsf{T}$ count in Clifford+T circuits
\iffull
in \appref{clifft}.
\else
in the Appendix.
\fi
%Our experiments demonstrate that our \algname{} algorithm,
%which guarantees local optimality,
%effectively uses \feyntool{} as an oracle for optimizing
%$\mathsf{T}$ count in Clifford+T circuits.
%
These experiments show that our \algname{} algorithm scales
to large circuits without reducing optimization quality.

\paragraph{Resynthesis methods.}
Resynthesis methods focus on decomposing unitaries into sequences of
smaller unitaries using algebraic structures of matrices.
Examples include
the Cartan decomposition \cite{tucci2005introduction},
the Cosine-Sine Decomposition (CSD),
the Khaneja Glaser Decomposition (KGD) \cite{khaneja2001cartan},
and the Quantum Shannon Decomposition (QSD).
Some synthesis methods demonstrate optimality for arbitrary unitaries of small size
(typically for fewer than five qubits),
particularly in terms of gate counts like CNOT gates \cite{rakyta2022approaching}.
However, their efficiency degrades significantly when dealing with larger unitaries;
furthermore, they require the time-consuming step of turning the circuit into a unitary.
QGo \cite{wu2020qgo} addresses this limitation with a hierarchical approach that
partitions and resynthesizes circuits block-by-block.
However, due to the lack of optimization across blocks,
the performance of QGo depends heavily on how circuits are partitioned.
Our local optimality technique, and specifically melding,
could be used to address this limitation of QGo.

\paragraph{Rule-based and peephole optimization methods.}
%Optimization methods often employ heuristic-based optimization
%to iteratively minimize the overall gate count.
%
%Their performance is often limited by the small set of rules used.
%Recent developments in optimization methods can be coarsely classified
%as rule-based, search-based, and learning-based techniques.
Rule-based methods find and substitute rules in quantum circuits to
optimize the circuit~\cite{iten2022exact, bandyopadhyay2020post,hietala2021verified, quartz-2022}.
VOQC~\cite{hietala2021verified} is
a formally verified optimizer that uses rules to optimize circuits.
VOQC implements several optimization passes inspired by state-of-the-art
unverified optimizer proposed by Nam et al.~\cite{Nam_2018}.
These passes include rules that perform $\mathsf{NOT}$ gate propagation, Hadamard gate reduction,
single- and two-qubit gate cancellation, and rotation merging.
Most of these passes take quadratic time in circuit size,
while some can take as much as cubic time~\cite{Nam_2018}.
Our experiments show that our local optimization algorithm \algname{}
effectively uses \voqc{} as an oracle for gate count optimization.
%

%
% UMUT: WE DO NOT IMPROVE RUNNING TIME OF VOQC
%Our results show that by focussing on local optimality,
%we can improve the running time of \voqc{} without any loss in circuit quality.

The notion of local optimality proposed in this paper is related
peephole optimization techniques from the classical compilers
literature~\cite{ct-compiler-2022,h+stratified-2016,sa-peephole-2006}.
Peephole optimizers typically optimize a small number of instructions,
e.g., rewriting a sequence of three addition operations into a single
multiplication operation.
Our notion of local optimality applies to segments of quantum
circuits, without making any assumption about segment sizes (in our
experiments, our segments typically contained over a thousand gates).
Because peephole optimizers typically operate on small instructions at
a time and because they traditionally consider the non-quantum
programs,  efficiency concerns are less important.
In our case, efficiency is crucial, because our segments can be large,
and optimizing quantum programs is expensive.
To ensure efficiency and quality, we devise a circuit
cutting-and-melding technique.

Prior work use peephole optimizers~\cite{prasad2006data,
liu2021relaxed} to improve the circuit one group of gates at a time,
and repeat the process from the start until they reach a fixed point.
The Quartz optimizer also uses a peephole optimization technique but
cannot make any quality guarantees~\cite{quartz-2022}.
Our algorithm differs from this prior work, in several aspects.
First,
it ensures efficiency, while also providing a quality guarantee based
on local optimality.
Key to attaining efficiency and quality is its use of circuit
cutting and melding techniques.
Second, our algorithm is generic: it can work on large segments (far
larger than a peephole) and optimizes each segment with an oracle,
which can optimize the circuit in any way it desires, e.g., it can use
any of the techniques described above.
%

\if0
%% UMUT: OUT OF PLACE
Although we focus on the Nam gate set for this evaluation, we note that
\quartz{} and \queso{} have shown competitive performance
for other gate sets (including IBM, Rigetti, and Ion)~\cite{quartz-2022, queso-2023}, and
we expect that this performance would translate to our setting.
\fi

\pyzx{}~\cite{kissinger2020Pyzx} is another rule-based optimizer
%However, it only minimizes $\mathsf{T}$ count and does not explicitly optimize gate count. We observe that \pyzx{} achieves worse performance on gate count, and can spend over 98\% of the time on circuit extraction for ZX-diagrams as opposed to optimization on ZX-diagrams for the sqrt circuit with 42 qubits.
that optimizes $\mathsf{T}$ count. It uses ZX-diagrams to optimize circuits and then extracts the circuit. Circuit extraction for ZX-diagrams is \#{}P-hard~\cite{de2022circuit}, and can take up much more time than optimization itself.
Because \algname{} invokes the optimizer many times, circuit extraction for ZX-diagrams can become a bottleneck. In addition, \pyzx{} only minimizes $\mathsf{T}$ count and does not explicitly optimize gate count. We therefore did not use \pyzx{} in our evaluation.
%Due to the uncertainty of the heuristics in the circuit extraction algorithm, it might not be efficient to pair \algname{} with \pyzx{} as an oracle for $\mathsf{T}$ count optimization.

\paragraph{Search-based methods.}
Rule-based optimizers may be limited by a small set of rules and are not exhaustive.
To address this,
researchers have developed search-based optimizers~\cite{queso-2023, quartz-2022, qfast,qsearch}
including \quartz{}~\cite{quartz-2022} and \queso{}~\cite{queso-2023}
that automatically synthesize exhaustive circuit equivalence rules.
Although their rule-synthesis approach differs,
both use similar algorithms for circuit optimization.
They iteratively operate on a search queue of candidate circuits.
In each iteration,
they pop a circuit from the queue,
rewrite parts of the circuit using equivalence rules,
and insert the new circuits back into the queue.
To manage the exponential growth of candidate circuits,
both tools use a ``beam search'' approach that limits the search queue size by
dropping circuits appropriately.
By limiting the size of the search queue,
\quartz{} and \queso{} ensure that
the space usage is linear relative to the size of the circuit.
Their running time remains exponential,
and they offer a timeout functionality,
allowing users to halt optimization after a set time.
This approach has delivered excellent reductions in gate count
for relatively small benchmarks~\cite{queso-2023, quartz-2022}.
However, for large circuits,
the optimizers do not scale well because they attempt to search an exponentially large search space.

QFast and QSearch apply numerical optimizations to
search for circuit decompositions that are close to the desired unitary~\cite{qfast, qsearch}.
Although faster than search-based methods \cite{davis2020towards},
these numerical methods tend to produce longer circuits, and their running time
is difficult to analyze.

%\paragraph{\feyntool{}.}

\if0
%% UMUT: REDUNDANT
The optimizers \quartz{} \cite{quartz-2022} and \queso{} \cite{queso-2023}
construct equivalent classes of quantum circuits and
use them to run a beam-search algorithm for reducing gate count (\secref{impl}).
%
These state-of-the-art optimizers are excellent for optimizing circuits
whose sizes are hundreds of gates.
\fi

\paragraph{Learning-based methods.}
Researchers have also developed machine learning models~\cite{fosel2021quantum}
for optimizing quantum circuits with variational/continuous parameters,
which reduce gate count by tuning parameters of
shallow circuit ansatze~\cite{mitarai2018quantum, ostaszewski2021reinforcement},
or by iteratively pruning gates \cite{sim2021adaptive, wang2022quantumnas}.
These approaches, however, are associated with substantial training costs~\cite{wang2022quantumnas}.

\section{Conclusion}
Quantum circuit optimization is a fundamental problem in
quantum computing.
State-of-the-art optimizers require at least quadratic time in the size of the
circuit, which does not scale to larger circuits that are necessary
for obtaining quantum advantage, and are unable to make strong quality
guarantees.
This paper defines a notion of local optimality and shows that it is
possible to optimize circuits for local optimality efficiently by
proposing a circuit cutting-and-melding technique.
With this cut-and-meld technique, the algorithm cuts a circuit into
subcircuits, optimizes them independently, and melds them efficiently,
while also guaranteeing optimization quality.
Our implementation and experiments show that the algorithm is
practical and performs well, leading to more than an order of magnitude
performance improvement (on average) while also improving optimization
quality.
These results show that local optimality can be effective in
optimizing large quantum circuits, which are necessary for quantum
advantage.
These results, however, do not suggest stopping to develop global
optimizers, which remains to be an important goal.
It is likely, however, due to inherent complexity of the problem (it
is QMA hard), global optimizers may struggle to scale to larger
circuits efficiently.
Because our approach to local optimality is generic, it can scale
global optimizers to larger circuits by employing them as oracles for
local optimization.

\bibliographystyle{plain}

% \bibliography{references}

\bibliography{Ref_QCS_Ding,local}

\iffull
\clearpage
\appendix
\if0
% moved back to main paper
\section{Detailed Results for RQ IV: Impact of compaction on the effectiveness of \algname{}}
\label{app:converge}
\figref{num-rounds} shows the number of rounds
and percentage optimizations for each round of our \algname{} algorithm.
%
Similar to \figref{main-gate-count},
we use our \algname{} optimizer, with \voqc{} as the oracle
and convergence ratio $\epsilon = 0.01$,
to optimize the benchmarks in the Nam gate set.
%
We observe that our algorithm converges quickly,
as it terminates after $2$ rounds of optimization.
%
Furthermore,
\textbf{the \algname{} algorithm
	consistently finds over 99\% of the optimizations in the first round itself.}
%
This is because our \algname{} ensures segment optimality directly
after the first round of optimization (see \secref{algorithm}).
%
Segment optimality guarantees that all segments of size $\Omega$ are optimal
with respect to the oracle;
the only difference between segment optimality and local optimality
is that segment optimality does not require compact circuits.
%
But as we observe that,
although compaction can enable some optimizations,
its impact on these benchmarks is minimal.
%
We also analyze these results for $\epsilon = 0$ and observe that
\algname{} produces locally optimal circuits quickly,
requiring $4$ rounds of optimization on average.
%

\input{fig/fig-compression-factor}
\fi

\section{Experiments with Clifford+T Gate Set}
\label{app:clifft}
In this section, we present the results of experiments with the
Clifford+T gate set. These results largely mirror the results
presented in the main body of the paper, showing similar efficiency
improvements and quality guarantees.

\subsection{RQ I: Effectiveness of \algname{} and local optimality}
We validate the effectiveness of local optimality
with the Clifford+T gate set. The Clifford+T gate set contains
the Hadamard ($\mathsf{H}$),
Phase ($\mathsf{S}$),
controlled-NOT ($\mathsf{CNOT}$),
and the $\mathsf{T}$ gate.
Our benchmark suite includes seven circuit families. We generate the Clifford+T circuits by transpiling the preprocessed Nam circuits
using Qiskit and gridsynth~\cite{gridsynth,qiskit-2023}.
For \texttt{bwt} and \texttt{hhl}, the two largest circuits of these
families failed to transpile (due to running out of memory) so we exclude
these from the Clifford+T evaluation.

For this gate set,
we run our \algname{} optimizer using \feyntool{}~\cite{amy2019formal} as the oracle optimizer
with segment depth $\Omega = 120$.
We describe in the next section, why we chose this value for $\Omega$. 
We evaluate the running time and output quality of our optimizer
against the baseline \feyntool{}.
The cost function is the $\mathsf{T}$ count, that is the number of $\mathsf{T}$ gates of the circuit.

We note that another possibility for optimizing Clifford+T circuits is
the \pyzx{} tool, which uses the ZX-diagrams for optimizations.
This tool, however, can require significant time to translate between
ZX-diagrams and the circuit representation.
For example, in our experiments, we found that for many circuits, the
\pyzx{} tool spends more than 50\% of its time on average translating
between circuits and diagrams.
Because our algorithm invokes the optimizer many times, the
translation times between circuits and diagrams can become a
bottleneck.
We therefore use the \feyntool{} in our evaluation, which does not
suffer from this problem, as it operates directly on the circuit.

\begin{figure}
	\centering\small
	\begin{tabular}{cccccccccc}
  & &  & \multicolumn{2}{c}{Time} & &  \multicolumn{2}{c}{T Count Reduction}  \\ \cmidrule(lr){4-5} \cmidrule(lr){7-8}
   Family & Qubits & Input T Count & \feyntool{} & \algname{} & \algname{} speedup & \feyntool{} & \algname{} \\

  \midrule\multirow{2}{*}{bwt}    & 17 & 169330 & 35730.0      & \textbf{354.8}  & 100.69 & -13.2\% & -13.2\% \\
                          & 21 & 214585 & 68301.4      & \textbf{569.1}  & 120.01 & -19.5\% & -19.4\% \\
  \midrule\multirow{4}{*}{grover} &  9 &   3927 & \textbf{2.6} & 3.3             &   0.79 & -31.2\% & -31.2\% \\
                          & 11 &  13720 & 12.0         & \textbf{10.5}   &   1.15 & -33.7\% & -33.7\% \\
                          & 13 &  36920 & 140.9        & \textbf{34.8}   &   4.05 & -33.1\% & -33.1\% \\
                          & 15 &  92016 & 2609.6       & \textbf{104.5}  &  24.97 & -32.7\% & -32.7\% \\
  \midrule\multirow{2}{*}{hhl}    &  7 &  61246 & 409.8        & \textbf{31.4}   &  13.03 & -31.2\% & -31.2\% \\
                          &  9 & 565183 & T.O.         & \textbf{535.4}  &  > 80.69 & T.O.    & -34.1\% \\
  \midrule\multirow{4}{*}{hwb}    &  8 &   5887 & \textbf{4.1} & 5.9             &   0.69 & -25.7\% & -25.7\% \\
                          & 10 &  29939 & 250.8        & \textbf{45.8}   &   5.48 & -29.8\% & -29.8\% \\
                          & 11 &  84196 & 4767.3       & \textbf{129.9}  &  36.7  & -31.3\% & -31.3\% \\
                          & 12 & 171465 & 21880.4      & \textbf{362.6}  &  60.35 & -34.1\% & -34.1\% \\
  \midrule\multirow{4}{*}{qft}    & 48 &  44803 & 195.1        & \textbf{77.1}   &   2.53 & -20.2\% & -20.2\% \\
                          & 64 &  61027 & 531.4        & \textbf{138.8}  &   3.83 & -20.3\% & -20.3\% \\
                          & 80 &  77251 & 1083.4       & \textbf{204.4}  &   5.3  & -20.3\% & -20.3\% \\
                          & 96 &  93475 & 1931.9       & \textbf{355.6}  &   5.43 & -20.3\% & -20.3\% \\
  \midrule\multirow{4}{*}{shor}   & 10 &   6104 & \textbf{2.0} & 4.0             &   0.51 & -19.7\% & -19.7\% \\
                          & 12 &  20180 & 21.0         & \textbf{16.5}   &   1.27 & -20.3\% & -20.3\% \\
                          & 14 &  70544 & 999.5        & \textbf{76.9}   &  12.99 & -20.5\% & -20.5\% \\
                          & 16 & 266060 & 28382.3      & \textbf{396.8}  &  71.53 & -20.6\% & -20.6\% \\
  \midrule\multirow{4}{*}{sqrt} & 42 &  25104 & 569.2        & \textbf{69.3}   &   8.21 & -37.4\% & -37.4\% \\
                          & 48 &  60366 & 5441.6       & \textbf{189.8}  &  28.67 & -39.9\% & -39.9\% \\
                          & 54 & 140830 & 36747.0      & \textbf{631.8}  &  58.16 & -41.7\% & -41.7\% \\
                          & 60 & 261308 & T.O.         & \textbf{1212.1} &  > 35.64 & T.O.    & -29.8\% \\

 \midrule
 \textbf{average} & & & & & > 9.91 & -27.1\% & -27.5\%
 \end{tabular}
	\caption{
		The figure shows the optimization results of optimizers
		$\algname{}$ and \feyntool{}, with $\mathsf{T}$ count as the cost function.
		It shows the running time in seconds for both optimizers (lower is better)
		and calculates the speedup of our \algname{} by taking the ratio of the
		two timings.
		The figure also shows the $\mathsf{T}$ count reductions of both tools.
		The results show that our $\algname{}$ delivers excellent time performance
		and runs almost an order of magnitude ($9.9\times$) faster than \feyntool{} on average.
		Our \algname{} optimizer achieves this speedup without any sacrifice in circuit quality,
		producing the same quality of circuits as \feyntool{}.
	}
	\label{fig:feynopt}
\end{figure}

\figref{feynopt} shows the results of this experiment.
The figure separates circuit families with horizontal lines
and sorts circuits within families by their size/number of qubits.
It presents the initial $\mathsf{T}$ count for each circuit
and the running times of both optimizers,
highlighting the fastest one in bold.
It computes the speedup achieved by our \algname{};
a speedup of $10\times$ means our optimizer runs $10\times$ faster.
The figure also
shows percentage reductions in $\mathsf{T}$ count achieved by
optimizers \feyntool{} and \algname{}.

The results show that our \algname{} generates high-quality circuits
for the Clifford+T gate set reasonably quickly, taking between $4$ seconds and approximately 20 minutes (for a circuit containing over 250,000 $\mathsf{T}$ gates).
%
\if0
and
does so in short running times,
including for circuits containing thousands to hundreds of thousands of $\mathsf{T}$ gates.
%
Our \algname{} optimizer optimizes
all circuits in times ranging from $4$ seconds to $2155$ seconds (< 1 hour).
%
\fi
%
We observe the following:
\begin{enumerate}
	\item \textbf{Time Performance:} Our \algname{} is faster for all circuit families,
	except perhaps for some smaller circuits.
	\item \textbf{Scalability: } Within each circuit family,
	the speedup of our \algname{} increases with increasing
	$\mathsf{T}$ counts. It is over $100\times$ faster for some cases and $9.9\times$ faster on average.
	\item \textbf{Circuit Quality:} Our \algname{} matches the $\mathsf{T}$ count reductions of \feyntool{} on
	all benchmarks.
\end{enumerate}

\myparagraph{Summary.}
\figref{feynopt}
shows that
our \algname{} optimizer can be significantly faster, especially for larger circuits, because it scales better, and does so without sacrificing optimization quality when optimizing for $\mathsf{T}$ count.
The experiment thus shows that the local optimality approach can work well  for Clifford+T circuits, especially for larger circuits.

\begin{figure}[t]
    \centering
    \includegraphics[width=\columnwidth]{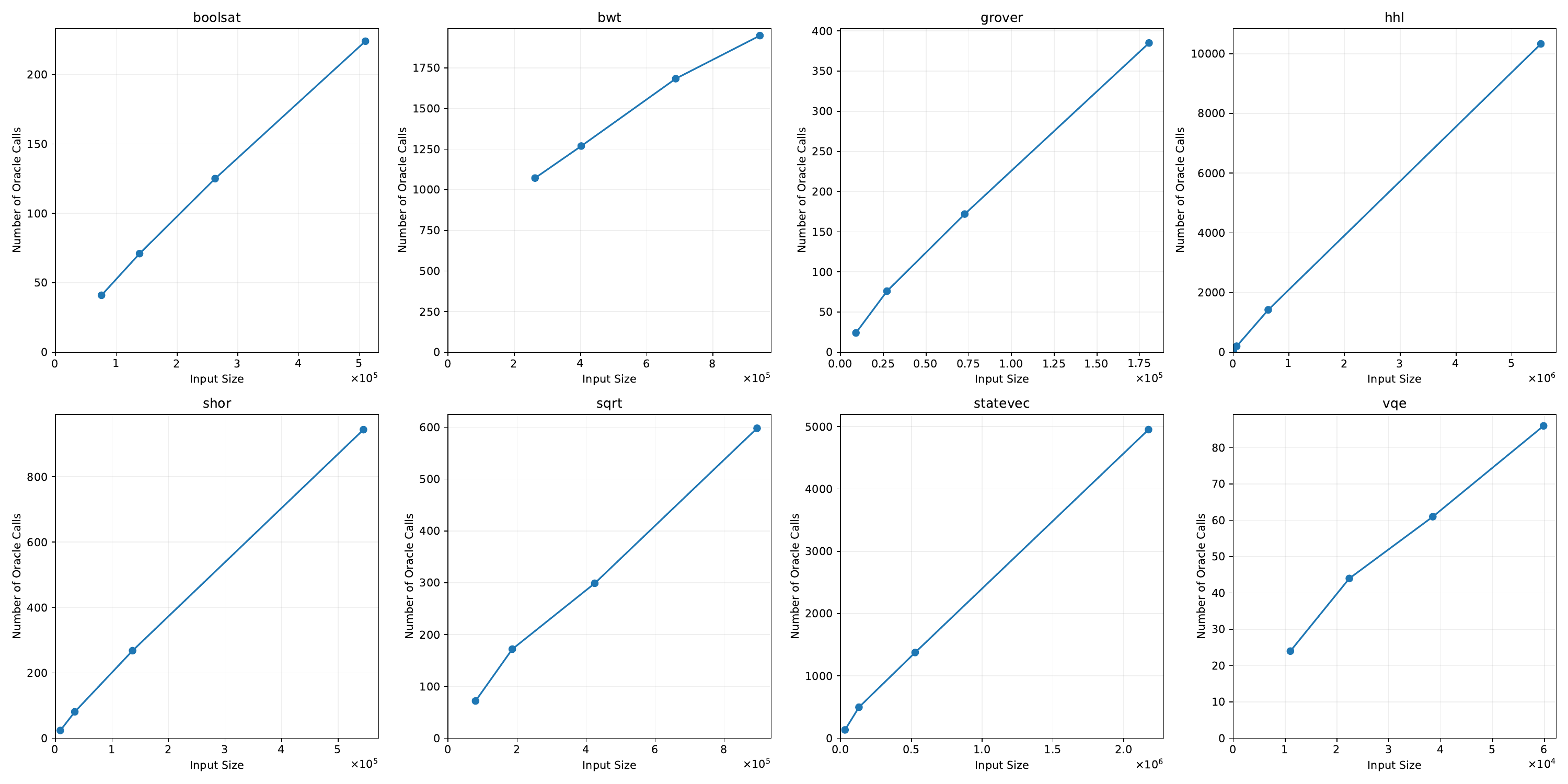}
    \vspace{-20pt}
    \caption{ The number of oracle calls versus input circuit size for
      all our circuits (Nam gate set).  The plots show that the number of
      calls scales linearly with the number of gates.  }
    \label{fig:plot-oracle-calls-all}
\end{figure}

\subsection{RQ IV: Impact of varying segment size $\Omega$}
\label{app:var-segment}
\begin{figure}[t]
	\centering
  % \begin{minipage}[b]{0.5\columnwidth}
  %   \includegraphics[width=\textwidth]{plots/omega_vs_quality.png}
  % \end{minipage}
  % \hspace{1cm}
  % \begin{minipage}[b]{0.5\columnwidth}
  %   \includegraphics[width=\textwidth]{plots/omega_vs_time.png}
  % \end{minipage}
  \includegraphics[width=\textwidth]{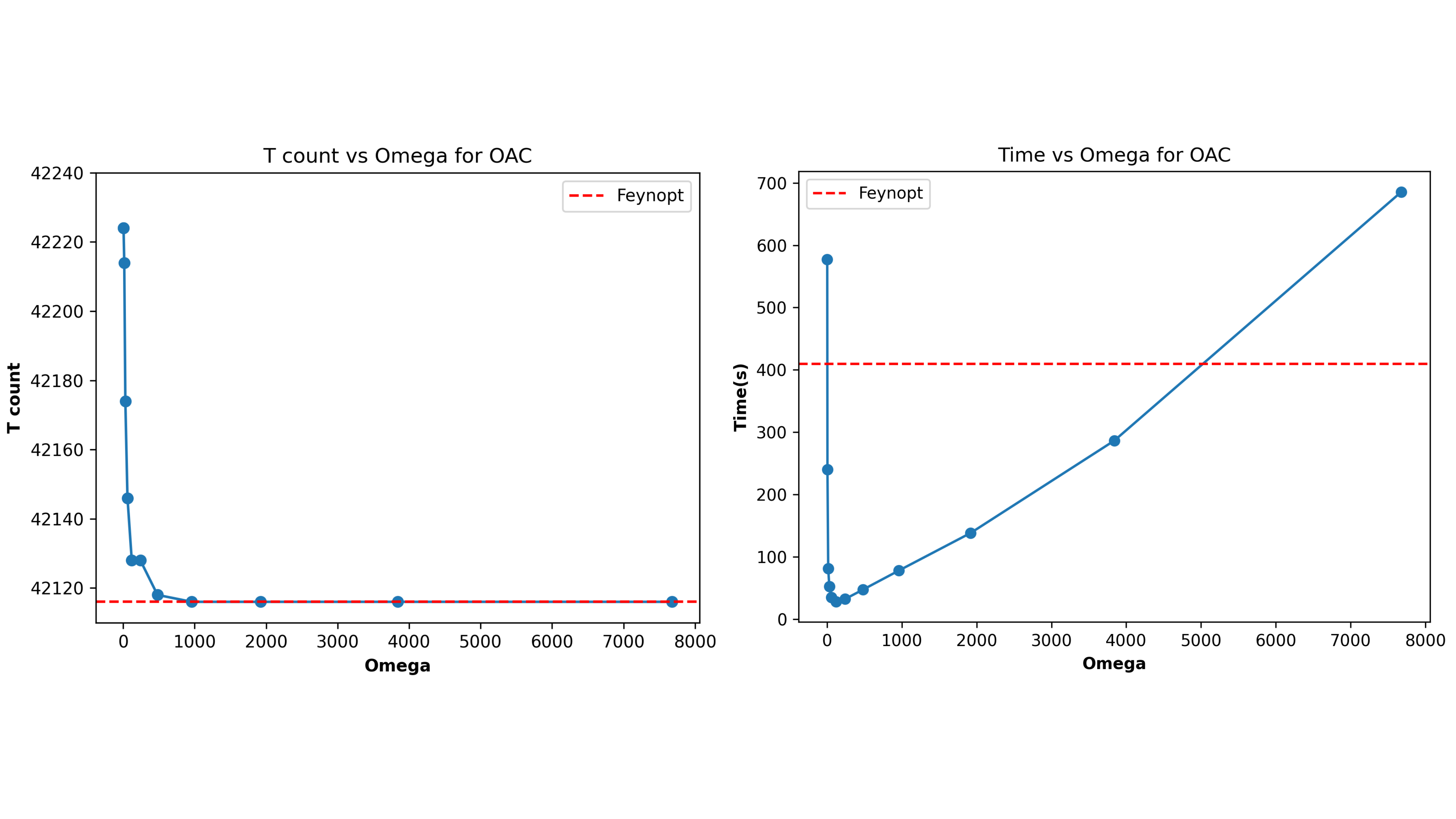}
    \vspace{-15mm}
	\caption{
    The figure plots the impact of the parameter $\Omega$ on
    the output T count (lower is better)
    and running time of $\algname{}$ optimizer on a hhl circuit with $7$ qubits.
    The dotted red lines in the plots denote the output T count and the running time of the oracle optimizer
    \feyntool{} on the whole circuit.
    For almost all values of $\Omega$, the output quality matches the oracle optimizer,
    demonstrating that local optimality is a robust quality criterion for T count optimization.
  }
    \label{fig:vary-segment}
\end{figure}
\begin{figure}[h]
  \centering
  \begin{tabular}{ccc}
      $\Omega$ & Output \textsf{T} count & Time (s) \\
      \midrule
      2 & 42224 & 577 \\
      5 & 42224 & 240.16 \\
      15 & 42214 & 81.10 \\
      30 & 42174 & 52.47 \\
      60 & 42146 & 34.58 \\
      120 & 42128 & 28.16 \\
      240 & 42128 & 32.08 \\
      480 & 42118 & 47.22 \\
      960 & 42116 & 77.85 \\
      1920 & 42116 & 138.16 \\
      3840 & 42116 & 286.31 \\
      7680 & 42116 & 685.69 \\
      \bottomrule
  \end{tabular}
  \caption{Results for \algname{} with \feyntool{} on the \texttt{hhl} circuit with $7$ qubits.}
  \label{fig:omega-clifft}
\end{figure}

We evaluate the impact of parameter $\Omega$
on the output quality and running time of our \algname{} algorithm.
We use $\algname{}$ on the Clifford+T gate set
and optimize for $\mathsf{T}$ count
with \feyntool{} as the oracle optimizer.
\figref{vary-segment} plots the output $\mathsf{T}$ count (number of $\mathsf{T}$ gates)
and the running time against $\Omega$.
The figure includes $\Omega$ values $2, 5, 15, 30, 60, 120 \dots 7680$;
we provide the raw data for the plot in \figref{omega-clifft}.
For the different values of $\Omega$,
we run our optimizer on the \texttt{hhl} circuit with $7$ qubits,
which initially contains 61246 $\mathsf{T}$ gates.
The red dotted line in the plots
show the output $\mathsf{T}$ count and running time of the
baseline, where the oracle \feyntool{} runs on the entire circuit.
The results show that for a wide range of $\Omega$ values,
our optimizer produces a circuit of similar quality to the baseline $\feyntool$
and typically does so in significantly less time
(with the exception of two values of $\Omega$ at the extremities: $2$ and $7680$).

\myparagraph{$\mathsf{T}$ count.}
The plot for $\mathsf{T}$ count shows that for small $\Omega$ (around $2$),
increasing it improves the output quality of \algname{}.
This is because the local optimality guarantee of $\algname{}$
strengthens with increasing $\Omega$,
as it guarantees larger segments to be optimal.
However, the benefits of increasing $\Omega$
become incremental around $\Omega = 60$,
where the $\mathsf{T}$ count stabilizes around 42140 (20 gates from optimal).
This demonstrates that local optimality is an effective quality criterion,
generating high-quality circuits
even with relatively small values of $\Omega$ (around $60$).

\myparagraph{Run time.}
One might expect the running time of \algname{}
to increase with segment size $\Omega$ because:
1) \algname{} uses the oracle on segments of size $2\Omega$,
thus each oracle call consumes more time
and 2) \algname{} gives a stronger quality
guarantee---larger the $\Omega$, stronger the guarantee given by local optimality.
Indeed, this intuition holds for most values in practice.
For $\Omega$ values ranging from $120$ to $7680$,
the running time increases with increasing $\Omega$.
In this range,
the increased cost of oracle call dominates the running time,
making it faster to partition circuits into smaller segments
and make many (smaller) calls to the oracle.

However, when $\Omega$ is very small,
we observe the opposite: increasing $\Omega$ reduces the running time.
\begin{wrapfigure}{r}{0.4\textwidth}
	\centering
	\includegraphics[width=0.39\textwidth]{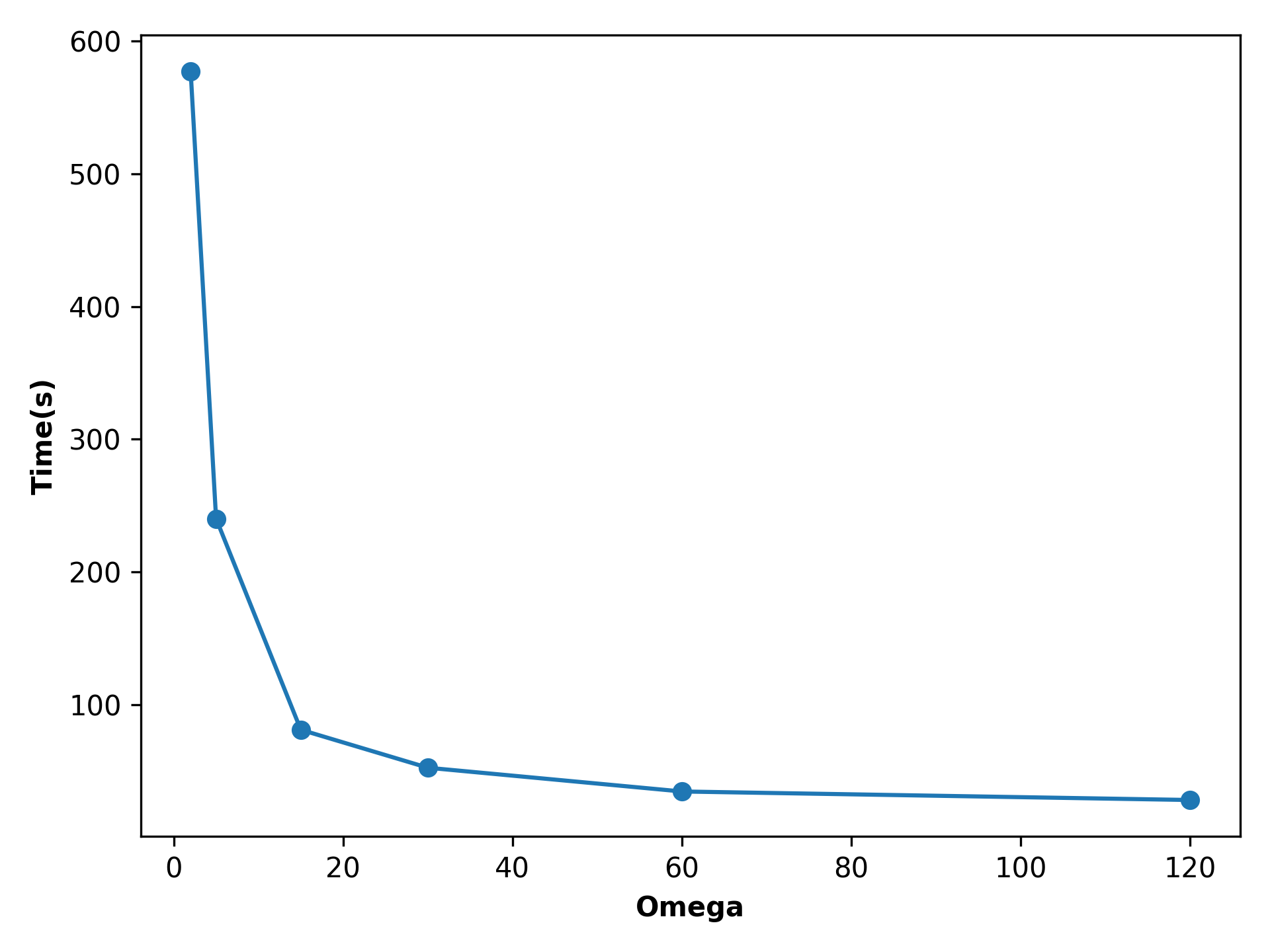}
	\caption{Zooming in: Time vs. Omega plot}
	\label{fig:zoom-plot}
\end{wrapfigure}
For reference,
we include \figref{zoom-plot},
which zooms in on the running time plot from \figref{vary-segment}
for initial values of $\Omega$, ranging from $2, 5, 15 \dots 120$.
For these smaller $\Omega$ values,
although each oracle call is fast,
the number of oracle calls dominates the time cost.
The \algname{} algorithm splits the circuit into a large number of small segments
and queries the oracle on each one, resulting in many calls to the optimizer.
When a circuit segment is small,
it is more efficient to directly call \feyntool{},
which optimizes it in one pass.
For this reason, $\Omega = 120$ is a good value for our optimizer $\algname{}$,
as it does not send large circuits to the oracle,
and also does not split the circuit into a large number of really small segments.

Overall,
we observe that for a wide range of $\Omega$ values,
our \algname{} outputs high quality circuits
and does so in a shorter running time than the baseline.
%

%%% ALGO

\section{Correctness of the Segment Optimality Algorithm}

\myparagraph{Segment Optimal Outputs.}
We prove that our \textsf{segopt} and \textsf{meld}
algorithms produce segment optimal circuits.
The challenge here is proving that
\textsf{meld} produces segment optimal outputs even though it may decide not to optimize one or both of the subcircuits based on the outcome of an optimization.
% %
%Once we establish the segment optimality of meld,
%the proof that \textsf{segopt} produces segment optimal outputs
%is straightforward.

\begin{lem}[Restatement of \lemref{meld-is-optimal}, Segment optimality of meld]
  %\label{lem:meld-is-optimal}
  Given any additive \cost{} function and any segment optimal circuits
  $C_1$ and $C_2$, the result of \lstinline{meld}$(C_1, C_2)$ is a
  segment optimal circuit $C$ and $\costof{C} \leq \costof{C_1} + \costof{C_2}$.
\end{lem}

\begin{proof}

  We show the lemma by induction on the total cost of the input, i.e.,
  $\costof {C_1} + \costof{C_2}$.
  In the base case, the input cost is zero, and therefore all segments of
  the input circuits have zero cost, due to additivity of the \cost{} function.
  The output of meld in this case is the concatenation $(C_1; C_2)$,
  because the oracle can not improve the boundary segment $W$.
  For the inductive step,
  we consider two cases, depending on whether the oracle optimizes the boundary segment $W$.

  \paragraph{Case I}
  When the oracle finds no improvement in $W$, the output is the concatenation $C = (C_1; C_2)$.
  %
  % In this case, the output cost is $\costof{C_1} + \costof{C_2}$ because
  % our cost function is additive w.r.t. concatenation.
  %
  To prove that the output $C$ is segment optimal,
  we show that every $\Omega$-segment of the output is optimal relative to
  the oracle.
  For any such $\Omega$-segment $X$, there are two cases.
  \begin{itemize}
  \item If $X$ is a segment of either $C_1$ or $C_2$, then $X$ is optimal relative to the oracle
  because the input circuits are segment optimal and $X$ is of length $\Omega$.

  \item If $X$ is not a segment of either $C_1$ or $C_2$, then $X$ is a sub-segment of $W$,
  because the segment $W$ has length $2\Omega$ and contains all possible $\Omega$-segments at the boundary of $C_1$
  and $C_2$. Because the oracle could not improve segment $W$, it can not improve its sub-segment $X$.
  \end{itemize}

  Thus, all $\Omega$-segments of the output are optimal relative to the oracle,
  and the output is segment optimal.

  \paragraph{Case II}
  If the oracle improves the boundary segment $W$,
  then the cost of the optimized segment $W'$ is less than cost of segment $W$, i.e.,
  $\costof{W'} < \costof{W}$.
  In this case, the function first melds the segment $C_1[0 : d_1 - \Omega]$ and the segment $W'$,
  by making a recursive call.
  We apply induction on the costs of the inputs to prove that the output circuit $M$ is segment optimal.
  Formally, we have the cost of inputs
  $\costof{C_1[0 : d_1 - \Omega]} + \costof{W'} < \costof{C_1[0 : d_1 - \Omega]} + \costof{W} \leq \costof{C_1} + \costof{C_2}$,
  because $W'$ has a lower cost than $W$.
  Therefore, by induction, the output $M$ is segment optimal
  and $\costof{M} \leq \costof{C_1[0 : d_1 - \Omega]} + \costof{W'}$.
  % Therefore, by induction, the first recursive call
  % to \lstinline{meld} returns an segment optimal output of size at most
  %

  For the second recursive call $\mathsf{meld}(M, C_2[\Omega : d_2])$,
  we can apply induction because $\costof{M} + \costof{C_2[\Omega : d_2]} < \costof{C_1} + \costof{C_2}$.
  This inequality follows from the cost bound on circuit $M$ above and using that
  $W'$ has a lower cost than $W$.
  Specifically,
  $\costof{M} + \costof{C_2[\Omega : d_2]} \leq \costof{C_1[0 : d_1 - \Omega]} + \costof{W'} + \costof{C_2[\Omega : d_2]}$,
   which is strictly less than
    $\costof{C_1[0 : d_1 - \Omega]} + \costof{W} + \costof{C_2[\Omega : d_2]} = \costof{C_1} + \costof{C_2}$.
  Therefore, by induction on the second recursive call,
  we get that meld returns a segment optimal circuit
  that is bounded in cost by $\costof{C_1} + \costof{C_2}$.
  Note that for both recursive calls, the induction relies on the fact that $W'$ has a lower cost than $W$,
  showing that the algorithm works because the recursion is tied to cost improvement.
\end{proof}

Based on \lemref{meld-is-optimal}, we can prove the following theorem.
%We leave the proof to
%\iffull
%\appref{thm-opt}.
%\else
%the Appendix.
%\fi
%We now prove that the hierarchical segment optimization via the
%divide-and-conquer technique always yields a segment optimal circuit.

\begin{theorem}[Restatement of \thmref{opt}, segment optimality algorithm] %\label{thm:opt}
  For any circuit $C$, the function $\mathsf{segopt}(C)$ outputs
  a segment optimal circuit.
  \end{theorem}

\begin{proof}
	We prove the theorem by induction on $d = \textsf{length}(C)$.
	In the base case, $d \leq 2\Omega$ and the circuit is fed to the oracle,
	thereby guaranteeing segment optimality.
	For the inductive case, the algorithm splits the circuit into
	two smaller circuits $C_1$ and $C_2$, and
	recursively optimizes them
	to obtain segment optimal outputs $C'_1$ and $C'_2$.
	The result $C = \mathsf{meld}(C_1', C_2')$ is then
	segment optimal, by \lemref{meld-is-optimal}.
	%
	% By induction, we have that both are segment optimal.
	% %
	% The algorithm then melds the two circuits, which by
	% Lemma~\ref{lem:meld-is-optimal} results in an segment optimal output.
	%
	% We thus conclude that the resulting output circuit is
	% segment optimal.
\end{proof}

\section{Efficiency of Segment Optimality Algorithm }
\label{sec:opt-analysis}

Having shown that the functions \textsf{segopt} and \textsf{meld}
produce segment optimal outputs,
we now analyze the runtime efficiency of these functions.
The runtime efficiency of both functions is data-dependent and varies with the number of optimizations found throughout the circuit.
To account for this,
we charge the running time to the cost improvement between the input and the output,
represented by $\Delta$.
We prove the following theorem.

\begin{theorem}[Restatement of \thmref{cost}, Efficiency of segment optimization]
  % Consider a segment length $\Omega$, an oracle that optimizes for the cost function $\mathbf{cost}$,
  % and a circuit $C$.
  The function $\mathsf{segopt}(C)$ calls the oracle at most
  $\mathsf{length}(C) + 2\Delta$ times on segments of length at most $2\Omega$,
  where $\Delta$ is the improvement in the cost of the output.
  \end{theorem}

The theorem shows that our segment optimality algorithm is productive in its use of the oracle.
Consider the terms, $\mathsf{length}(C)$ and $2\Delta$, in the bound on the number of oracle calls:
The first term, $\mathsf{length}(C)$, is for checking segment
optimality---even if the input circuit is already optimal,
the algorithm needs to call the oracle on each segment and confirm it.
The second term $2\Delta$,
shows that all further calls result from optimizations,
with each optimization requiring up to two oracle calls in the worst case.
This shows that our algorithm (alongside meld)
carefully tracks the segments on which the oracle needs to be called
and avoids unnecessary calls.

The theorem also highlights one of the key features of our $\mathsf{segopt}$ function:
it only uses the oracle on small circuit segments of length $2\Omega$.
Suppose $q$ is the number of qubits in the circuit.
Then the $\mathsf{segopt}$ function only calls the oracle on
manageable circuits of size less than or equal to $q * 2\Omega$,
which is significantly smaller than circuit size.
This has performance impacts because, for many oracles,
it is faster to invoke the oracle many times on small circuits
rather than invoking the oracle a single time on the full circuit (see \secref{eval}).

Since we bound number of oracle calls in terms of the cost improvement $\Delta$,
an interesting question is how large can $\Delta$ be?
When optimizing for gate-counting metrics
such as $\mathsf{T}$ count (number of $\mathsf{T}$ gates), $\mathsf{CNOT}$ count (number of $\mathsf{CNOT}$ gates),
gate count (total number of gates),
the cost improvement is trivially bounded by the circuit size,
i.e., $\Delta \leq \sizeof{C}$.
This observation shows that our \textsf{segopt}
function only makes a linear number of calls (with depth and gate count)
to the oracle.

\begin{corollary}[Restatement of \corref{linear-calls}, Linear calls to the oracle]
  When optimizing for gate count, our $\mathsf{segopt}(C)$ makes a linear,
  $O(\mathsf{length}(C) + \sizeof{C})$, number of calls to the oracle.
\end{corollary}

We experimentally validate this corollary in \secref{eval-calls},
where we study the number of oracle calls made by our algorithm
for many circuits.
The crux of the proofs of \thmref{cost} and the corollary 
is bounding the number of calls in the \textsf{meld} function,
because it can propagate optimizations through the circuit.
We show the proof below.
%We state bound below and leave the proofs to 
%\iffull
%\appref{thm-meld-num-oracles}.
%\else
%the Appendix.
%\fi
%
\begin{lem}[Restatement of \lemref{meld-num-oracles}]
	%\label{lem:meld-num-oracles}
	Computing $C =$ \lstinline{meld}$(C_1, C_2)$
	makes at most $1 + 2\Delta$ calls to the oracle,
	where $\Delta$ is the improvement in cost,
	i.e., $\Delta = \left(\costof {C_1} + \costof{C_2} - \costof{C}\right)$.
\end{lem}
\begin{proof}
	We prove this by induction on the input cost $\costof {C_1} + \costof{C_2}$.
	The base case is straightforward: when the input cost is zero,
	the meld function only makes the one call and returns.
	For the inductive step,
	we consider two cases depending on whether the oracle improves the boundary
	segment $W$.
	
	\paragraph{Case I} If the oracle finds no improvement in the segment $W$, then
	there is exactly one call to the oracle.
	Thus we have $1 \leq 1 + 2\left(\costof {C_1} + \costof{C_2} - \costof{C}\right)$.
	
	\paragraph{Case II}
	When the oracle finds an improvement in the segment $W$,
	it returns $W'$ such that $\costof{W'} \leq \costof W - 1$.
	(There is a difference of at least 1 because $\cost{}(-) \in \mathbb{N}$.)
	In this case, the meld function recurs twice,
	first to compute $M = \mathsf{meld}(C_1[0 : d_1 - \Omega], W')$,
	and then to compute the output $C = \mathsf{meld}(M, C_2[\Omega:d_2])$.
	The total number of calls to the oracle can be decomposed as follows:
	
	\begin{itemize}
		\item 1 call to the oracle for the segment $W$.
		\item Inductively,
		at most $1 + 2\left(\costof {C_1[0 : d_1 - \Omega]} + \costof{W'} - \costof{M}\right)$ calls to the oracle for the first meld.
		\item Inductively, at most
		$1 + 2\left(\costof M + \costof{C_2[\Omega:d_2]} - \costof{C}\right)$ calls to the oracle in the second recursive meld.
	\end{itemize}
	Adding these up yields at most
	$1 + 2\left(\costof {C_1} + \costof{C_2} - \costof{C}\right)$ calls to the oracle, as desired.
	(We use here the facts
	$\costof{W'} \leq \costof W - 1$ and
	$\costof {C_1[0 : d_1 - \Omega]}  + \costof W + \costof {C_2[\Omega:d_2]} = \costof {C_1} + \costof{C_2}$.)
\end{proof}

% \begin{corollary}
%   Given an optimizer that takes $T(q, \sizeof{C})$ time on a circuit $C$ with $q$ qubits and size \sizeof{C},
%   computing $C' = \mathsf{segopt}(C)$ requires at most $O(T(q, \sizeof{C}) * (\mathsf{length}(C) + \Delta))$ time,
%   where $\Delta = \costof{C} - \costof{C'}$ is the improvement in cost.
% \end{corollary}

% \begin{corollary}
%   Given an optimizer that takes $T(q, \mathsf{length}(C))$ time on a circuit $C$ with $q$ qubits,
%   computing $C' = \mathsf{segopt}(C)$ takes $O((\mathsf{length}(C) + 2\Delta) * T (q, 2\Omega))$ time
% \label{cor:linear-calls}
% \end{corollary}

% The corollary demonstrates that when optimizing for gate count,
% our \textsf{segopt} function scales linearly with both the number of gates and the length/depth
% of the circuit.
% %
% Furthermore, our \textsf{segopt} function and the oracle
% exhibit similar scaling behavior in the number of qubits.

% The corollary shows that when optimizing for gate count,
% our \textsf{segopt} function scales linearly
% with the number of gates and depth of the circuit.
% %
% Both the oracle and our \textsf{segopt} function scale similarly with the number of qubits.

% %

\begin{lem}[Bounded calls to oracle in meld]
  \label{lem:meld-num-oracles}
Computing $C =$ \lstinline{meld}$(C_1, C_2)$
makes at most $1 + 2\Delta$ calls to the oracle,
where $\Delta$ is the improvement in cost,
i.e., $\Delta = \left(\costof {C_1} + \costof{C_2} - \costof{C}\right)$.
\end{lem}

The \algname{} algorithm
uses the segment optimality guarantee, given by \textsf{segopt},
to produce a locally optimal circuit.
In \figref{lopt-code}, the function \algname{} takes the input circuit $C$
and computes the circuit $C' = \mathsf{segopt}(\mathsf{compact}(C))$.
It then checks if $C'$ differs from the input $C$ and if so,
it recurses on $C'$.
% It then checks if $C'$ differs from the input $C$.
% %
% If $C'$ differs from $C$---which could be due to optimization, or compaction, or
% both---then \algname{} continues recursively on $C'$.
% %
% Otherwise, if $C'$ and $C$ are identical,
% it returns the circuit $C'$ because no optimization or compaction occured.
The function repeats this until convergence, i.e.,
until the circuit does not change.
In the process, it
computes a sequence of segment optimal circuits $\{C_i\}_{1 \leq i \leq \kappa}$,
where $\kappa$ is the number of rounds until convergence:
\[C \to C_1 \to C_2 \to \ldots \to C_\kappa \quad\text{where}\quad
\begin{array}{c}
\costof{C_i} > \costof{C_{i+1}} \\
\windowopt\Omega{C_i} \\
\locallyopt{\Omega}{C_\kappa}
\end{array}
\]
Each intermediate circuit $C_i$ is segment optimal
and gets compacted before optimization in the next round.
Given that \algname{} continues until convergence,
the final circuit, $C_\kappa$, is locally optimal.

% This argument subtly hinges upon a property of our segment optimality algorithm, in particular,
% that $\mathsf{segopt}(C)$ differs from $C$ only if the cost has improved.
% %
% The function $\mathsf{segopt}$ and its key subroutine, meld, are careful to ensure this
% by guarding the oracle calls:
% %
% if applying the oracle does not improve the cost, then the output of the oracle is discarded
% (see lines \ref{line:guard1start}-\ref{line:guard1stop} and \ref{line:guard2start}-\ref{line:guard2stop}
% in \figref{lopt-code}).
% %
% (This aligns with the rewriting semantics of \secref{lang}, which only performs optimization
% rewrites if the cost decreases.)
% %

Using \thmref{cost},
we observe that the number of oracle calls in each round is bounded by
$\mathsf{length}(C_i) + 2\Delta_i$,
where $\Delta_i$ is the improvement in cost (i.e., $\Delta_i = \costof{C_i} - \costof{C_{i+1}}$).
Thus, the total number of oracle calls performed by \algname{} can be bounded by
$O(\Delta + \mathsf{length}(C) + \sum_{1 \leq i \leq \kappa} \mathsf{length}(C_i))$
where $\Delta = \costof{C} - \costof{C_\kappa}$ is the overall improvement in cost between the input and output.

\begin{theorem}
\label{thm:local-opt}
Given an additive cost function,
the output of $\algname{}(C)$ is locally optimal
and requires $O(\Delta + \mathsf{length}(C) + \sum_{1 \leq i \leq \kappa} \mathsf{length}(C_i))$ oracle calls,
where $C_i$ denotes the circuit after $i$ rounds,
$\kappa$ is the number of rounds,
and $\Delta$ is the end-to-end cost improvement.
\end{theorem}

The bound on the number of oracle calls in \thmref{local-opt} is very general
because it applies to any oracle (and gate set).
For particular cost functions, we can use it to deduce a more precise bound.
Specifically, in the case of gate count, where
$\costof C = \sizeof C$, we get the following bound.

\begin{corollary}\label{cor:oac-linear-calls}
When optimizing for gate count,
$\algname{}(C)$ performs at most $O(\sizeof{C} + \kappa \cdot \mathsf{length}(C))$ oracle calls.
\end{corollary}

An interesting question is whether or not it is possible to bound the number of rounds, $\kappa$.
In general, this will depend on a number of factors,
such as the quality of the oracle and how quickly the optimizations it performs converge
across compactions.
In practice (\secref{eval}), we find that the number of rounds required is small, and that
nearly all optimizations are performed in the first round itself.
For gate count optimizations in particular, Corollary~\ref{cor:oac-linear-calls} tells us
that when only a small number of rounds are required, the scalability of \algname{} is
effectively linear in circuit size.

\if0
\subsection{Worst-case number of rounds for \algname{}}
\input{fig/fig_quadratic_lower_bound}
%
\figref{quadratic_lower_bound} gives an example showing that a linear number of compactions is necessary in the worst case.
%
The example uses $\Omega=3$.
%
We optimize the solid green segment from (a) to (b) by canceling the two CNOTs.
%
However, at (b), no more optimization can be applied.
%
Although the dashed green segment is not optimal, the two CNOTs are separated enough so that we cannot optimize them.
%
So, we must apply a compaction step to bring the two CNOTs together and then optimize them, as shown in (c).
%
This pattern will continue and can require a number of compactions proportional to the number of qubits.
\fi

\subsection{Circuit Representation and Compaction}
%\ur{Maybe move to this to the end.  Say, for the analysis, we assume that split and concat are constant.  We describe in section x, how to achieve this by using a specific circuit representation.
%
%UPDATE: I moved this here.  But I am not sure where it is used.  We
%seem to be proving only the number of oracle calls and not end-to-end
%run time.
%
%}

Our algorithm represents circuits using a hybrid data structure that
switches between sequences and linked lists for splitting and
concatenating circuits in $O(1)$ time.
Initially, it represents the circuit as a sequence of layers,
enabling circuit splits in $O(1)$ time.
Later during optimization
when circuit concatenation is required,
the algorithm switches to a linked list of layers,
enabling $O(1)$ concatenation.
At the end of each optimization round,
the algorithm uses a function $\mathsf{compact}(C)$
to revert to the sequence of layers representation.

The function $\mathsf{compact}(C)$ ensures that the resulting circuit
is equivalent to $C$ while eliminating any unnecessary ``gaps'' in the layering,
meaning that every gate has been shifted left as far as possible.
For brevity, we omit the implementation of the $\mathsf{compact}$ function from \figref{lopt-code}.
Our implementation in practice is straightforward: we use a single left-to-right
pass over the input circuit and build the output by iteratively adding gates.
The time complexity is linear, i.e., $\mathsf{compact}(C)$ requires $O(|C| + \mathsf{length}(C))$ time.
%
%
% Local optimality requires that a layered circuit be as compact as possible,
% and the function $\mathsf{compact}(C)$ satisfies this condition.
%
%

% %
% In particular, if $C' = \mathsf{compact}(C)$ then we have $\compressed{C'}$, using the
% definition of compactness from \secref{lang}; the resulting circuit $C'$ is identical
% to $C$ except that every gate has been shifted left as far as possible.
%
% Our \algname{} algorithm uses this $\mathsf{compact}$ function in two ways.
% \begin{enumerate}
%     \item The output of every \oracle{} call is compacted. The cost of this compaction can be charged
%     against the oracle, because compaction is no more expensive than reading and writing the
%     input and output circuits. In this sense, compacting the output of every oracle call is ``free''
%     from an efficiency perspective.
%     Note that these compactions are not essential for any of our theorems; we include these in the
%     algorithm description because they are useful in practice, helping ensure that every optimization
%     phase is as productive as possible.
%     \item After every optimization phase, the entire circuit is compacted. Note that these compactions
%     are essential for guaranteeing local optimality.
% \end{enumerate}

%%% ALGO

\section{Proof of Termination}
\label{app:lem-pot-dec}
We prove below (\lemref{pot-dec}) that the potential decreases on every step.
\thmref{convergence} then follows from \lemref{pot-dec}, because ordering by $\Phi$
is well-founded and cannot infinitely decrease.
%
% \thmref{convergence} by induction over circuits ordered by $\Phi$.
% and because the range of our potential function is well-founded,
% there can be no infinite rewriting sequences.

\begin{lemma}
	\label{lem:pot-dec}
	For any additive function $\cost{}: \textit{Circuit} \to \mathbb{N}$,
	% and any well-formed circuit $C$,
	if $\localstep\Omega{C}{C'}$ then $\Phi(C') < \Phi(C)$.
\end{lemma}
\begin{proof}
	% If $\locallyopt\Omega{C}$, then $C = C^\text{OPT}$ and we are done; otherwise,
	% by \lemref{can-step} we have $\localstep\Omega{C}{C'}$.
	%
	There are two cases
	for $\localstep\Omega{C}{C'}$:
	either \rulename{Lopt} or
	\rulename{ShiftLeft}.
	In each case we show $\Phi(C') < \Phi(C)$.
	
	In the case of \rulename{Lopt}, we have $\localstep\Omega C {C'}$ where
	$C = (P ; C'' ; S)$ and $C' = (P ; \oracle{}(C'') ; S)$.
	In $C'$, the segment $C''$ has been improved by one call to the \oracle{},
	i.e., $\costof{\oracle{}(C'')} < \costof{C''}$.
	Because the \cost{} function is additive, we have that
	$\costof{C'} = \costof{P ; \oracle{}(C'') ; S} < \costof{P ; C'' ; S} = \costof{C}$.
	This in turn implies $\Phi(C') < \Phi(C)$ due to lexicographic ordering on the
	potential function.
	
	% a segment of $C$ is improved by one
	% call to the \oracle{}.
	% %
	% Because the \cost{} function is additively monotonic, we therefore have
	% that $\costof{C'} < \costof{C}$, which in turn implies that
	% $\Phi(C') < \Phi(C)$ due to lexicographic ordering.
	
	% To show that the rewrite system terminates on any given circuit $C$,
	% we define a well-founded order $\prec$ on the set of circuits
	% and show that each rewrite descends down the order,
	% i.e., for each $\localstep\Omega{C}{C'}$, $C' \prec C$.
	% %
	% Because a well-founded order can not have infinitely descending chains,
	% the number of possible rewrites starting from the circuit $C$ are finite,
	% showing that the circuit $C$ can be made locally optimal in a finite number of steps.
	% %
	
	In the case of \rulename{ShiftLeft}, we have $\localstep\Omega C {C'}$ where
	$C = (P ; \langle L_1, L_2 \rangle ; S)$ and $C' = (P ; \langle L_1', L_2' \rangle ; S)$
	and $L_1' = L_1 \cup \{G\}$ and $L_2' = L_2 \setminus \{G\}$.
	Because the cost function is additive, we have $\costof{\langle L_1, L_2 \rangle} = \costof{\langle L_1', L_2' \rangle}$
	and therefore $\costof C = \costof {C'}$.
	To show $\Phi(C') < \Phi(C)$, due to the lexicographic ordering, it remains to show
	$\mathsf{IndexSum}(C') < \mathsf{IndexSum}(C)$.
	This in turn follows from the definition of $\mathsf{IndexSum}$; in particular, plugging
	in $|L_1'| = |L_1| + 1$ and $|L_2'| = |L_2| - 1$ we get
	$\mathsf{IndexSum}(C') = \mathsf{IndexSum}(C) - 1$.
	Thus we have $\Phi(C') < \Phi(C)$.
\end{proof}

\fi

\end{document}